\documentclass[11pt,a4paper]{amsart}
\usepackage{amsfonts,amsmath,amsthm,amssymb,url,verbatim,multirow,color}
\newcommand\E{\ensuremath{\mathbb{E}}}
\newcommand\PP{\ensuremath{\mathbb{P}}}
\renewcommand\O{\ensuremath{\mathcal{O}}}

\newcommand\e{\ensuremath{\mathrm{e}}}
\newcommand\la{{\lambda}}
\newcommand\eps{{\epsilon}}

\newcommand\bt{\boldsymbol{\Theta}}

\newcommand\half{\ensuremath{\frac{1}{2}}}
\newcommand\F{\ensuremath{\mathcal{F}}}
\newcommand\R{\ensuremath{\mathbb{R}}}

\DeclareMathOperator*{\argmax}{arg\, max}

\newtheorem{thm}{Theorem}
\newtheorem{lemma}{Lemma}
\newtheorem{prop}{Proposition}

\theoremstyle{remark}
\newtheorem{rem}{Remark}

\usepackage[margin=1in]{geometry}

\usepackage{natbib}

\usepackage{ifpdf}
\ifpdf
\usepackage[pdftex]{graphicx}
\else
\usepackage[dvips]{graphicx}
\fi

\title{Illiquidity Effects in Optimal Consumption-Investment Problems}

\author[M. Ludkovski]{Michael Ludkovski$^\dag$} \address{Department of Statistics and Applied Probability,
University of California, Santa Barbara, CA 93106, USA} \email{ludkovski@pstat.ucsb.edu,
www.pstat.ucsb.edu/faculty/ludkovski}

\author[H. Min]{Hyekyung Min} \address{Department of Mathematics, Case Western Reserve University, Cleveland, OH 44106, USA} \email{hxm182@case.edu}

\thanks{$^\dag$ Corresponding Author. email: ludkovski@pstat.ucsb.edu, tel: 1(805)893-5634. \\ Address: Department of Statistics and Applied Probability,
University of California, Santa Barbara, CA 93106, USA}


\begin{document}
\begin{abstract}
We study the effect of liquidity freezes on an economic agent optimizing her utility of consumption in a perturbed Black-Scholes-Merton model. The single risky asset follows a geometric Brownian motion but is subject to liquidity shocks, during which no trading is possible and stock dynamics are modified. The liquidity regime is governed by a two-state Markov chain. We derive the asymptotic effect of such freezes on optimal consumption and investment schedules in the two cases of (i) small probability of liquidity shock; (ii) fast-scale liquidity regime switching.
Explicit formulas are obtained for logarithmic and hyperbolic utility maximizers on infinite horizon. We also derive the corresponding loss in utility and compare with a recent related finite-horizon model of Diesinger, Kraft and Seifried (2010).
\end{abstract}

\maketitle

\section{Introduction}
The theory of optimal investment has been one of the cornerstones of the modern mathematical finance. Working in the context of the Black-Scholes model of a single risky asset,
\cite{Merton71}  in
his seminal work obtained closed-form expressions for optimal investment and consumption strategies of an agent with power-utility preferences. These results relied on the fundamental assumptions of frictionless and continuous trading.  In the past thirty five years, extensive efforts have been devoted to reconciling this assumption with real financial markets. A large body of literature (see e.g.\ the textbook
\cite{KarShreveMethods} and references therein) treats optimal investment with (either proportional or fixed) transaction costs; another significant strand addresses optimal investment with discrete trading times \citep{Rogers00,RogersZane,PhamTankov08MF,Matsumoto,TankovPhamGozzi}. Finally, several authors considered optimal investment-consumption models under stochastic volatility \citep{FlemingHernandez03,FPS-Book,BauerleRieder04}.

Recently, the issue of \emph{illiquid} financial markets has been considered. From one direction, several authors \citep{RogersCetin07,CetinJarrowProtter04} have analyzed optimal investment by a large trader whose actions affect the dynamics of the market. From a different perspective, \cite{Diesinger} have treated markets where trading may be interrupted. The latter phenomenon means that during times of crisis markets freeze and often no trading is possible at all, no matter the price. Thus, during major financial dislocations, trading halts and circuit-breakers are implemented to alleviate panic. These shutdowns typically follow a major market drop and are often triggered by the index losing a preset percent of its value within a single trading session. For instance, on October 28, 1997 the New York Stock Exchange (NYSE) experienced two circuit-breakers once the Dow Jones index lost 5\% and then 7.5\% percent of its value. In each case, trading was halted for an hour.
More recently, during the 2008 financial crisis, the Russian RSX exchange was closed by government decree for three days during Sep 20-23, 2008. Many other emerging market bourses in Brazil, India, Thailand, Bangladesh, etc., also faced repeated halts for parts or all of a trading session in September-October 2008. Similarly, after the 9/11 attacks in 2001, the NYSE was closed for four days. On many of the smaller national exchanges, investors must contend with regular trading halts due to political, regulatory and financial upheavals.


More lengthy interruptions have been witnessed during times of war. In the 1940s, the Zurich stock exchange closed for over a month in May 1940; the Frankfurt stock exchange was closed for 5 months in 1945. At the onset of World War I the NYSE was closed from July 31, 1914 until November 28, 1914 \citep{Diesinger}. Moving away from stocks, lengthy liquidity freezes are commonplace in more exotic asset classes. The recent credit crisis is a case in point; 2007-2008 saw major liquidity crises in US sub-prime mortgage backed securities (MBS), Canadian asset backed commercial paper (ABCP) and worldwide collaterized debt obligations (CDOs). The collapse of Lehman Brothers in September 2008 led to several months of virtual shutdown in trading of credit products and there were many anectodal stories of investors unable to sell their ``toxic'' assets at any price.

The effect of such liquidity crises on optimal investment is threefold. First, inability to rebalance the portfolio during the interruption reduces the long-run utility of the agent. Second, liquidity crises are commonly accompanied by major market drops (cf.\ Russia  in 2008, US in 2001, Germany in 1945) which lead to significant unavoidable losses for the agent. Third, during a trading interruption, agents are unable to liquidate their holdings to finance consumption. Thus, agents may experience shortage of funds and are forced to dramatically curtail their consumptions while their investments are frozen. This was widely documented in the mass media in autumn of 2008 and was one of the major mechanisms of the transmission of the financial crisis into the real economy.

In this paper we investigate the effect of such trading interruptions on the optimal investment and consumption strategies of a power-utility agent. We maintain the Black-Scholes model for the risky asset and the assumption of frictionless trading outside of the liquidity event. By keeping our setup identical to the classical Merton model except for the feature of liquidity freezes, we are able to draw clear comparisons to the original benchmark.

Because historically major liquidity crises have been \emph{rare events}, we undertake asymptotic analysis of the effect of illiquidity on optimal strategies. Indeed, due to their rarity, statistical estimation of the frequency of liquidity crises is fraught with difficulties and the corresponding model parameter cannot be well calibrated. One of our contributions are closed-form formulas for the first-order corrections to the classical Merton solutions for power-utility agents. These formulas show the impact of the liquidity crises frequency and duration and provide simple rules of thumb for the investors. Similar asymptotic expansions for other models of liquidity were undertaken in \cite{RogersZane} and \cite{Matsumoto}.


An alternative viewpoint is that liquidity crises are commonplace but have a very short duration. Mathematically, this means that liquidity is a fast-scale process, modulating the slow-scale geometric Brownian motion model of the risky asset. Such theory of fast-scale hybrid diffusions was developed by \cite{IlinYinJOTA99,IlinYinJMAA99}. The effect of fast-scale stochastic volatility on optimal investment-consumption was studied by \cite{BauerleRieder04} in a regime-switching volatility model and by \cite{FPS-Book} in a diffusion setting. The key observation is that the singular perturbation of the HJB equations requires averaging of the Sharpe ratios rather than of the volatilities themselves. Our results below (fully explicit for the logarithmic and hyperbolic utilities) complement this analysis; in particular our perturbation is of a new type since the trading control is only active in one regime.
On a practical level, this homogenized model allows investors to adjust their strategies to account for a fast-scale friction in the market.

Our model is closest to the recent work of \cite{Diesinger}. In the latter work, which was one of the inspirations for our research, the authors consider a model with similar motivations and techniques. Following their setup, we treat liquidity as a two-state Markov chain which modulates the market dynamics and imposes constraints on the trading strategies.
While \cite{Diesinger} work on a finite horizon and study maximization of terminal utility, we instead focus on an optimal consumption model on infinite horizon.
Also, \cite{Diesinger} do not carry out any asymptotic analysis and do not treat fast-scale liquidity. An earlier study of illiquidity effects was carried out by \cite{SchwartzTebaldi} who considered optimal investment where the liquidity shock times are known in advance. Finally, our use of a finite-state Markov chain to modulate the investment regimes also resonates with regime-switching models of optimal investment, see \cite{SotomayorCadenillas09,Zariphopoulou92}.

The rest of the paper is organized as follows. Section \ref{sec:market} gives a formal description of our hybrid diffusion market model and the optimization problem for the investor. Section \ref{sec:log-utility} provides the solution for a log-utility maximizing investor. Section \ref{sec:power-utility} repeats the analysis for a general investor with HARA utility; in particular Section \ref{sec:hyperbolic-utility} provides explicit solutions for the case of hyperbolic utility; also implicit formulas are obtained for square-root and inverse-square-root utilities. All these analytical results are illustrated with several numerical examples in Section \ref{sec:numeric}. In Section \ref{sec:homogenized} we study the fast-scale limit of our illiquidity model and its effect on optimal investment and consumption. Section \ref{sec:finite-horizon} extends to a finite-horizon setting where we consider a log-utility maximizing agent, as well as provide the asymptotic analysis of the terminal log-utility maximizer of \cite{Diesinger}.  Finally, Section \ref{sec:conclusion} concludes.

\section{Market with Regime-Switching Liquidity}\label{sec:market}
Let $(M_t)$ be a two-state continuous Markov chain with infinitesimal generator
$$
Q = \begin{pmatrix} -\lambda_{01} & \la_{01} \\ \la_{10} & -\la_{10} \end{pmatrix}.
$$
Associated with $M$ are two Poisson processes $N_{01}$ and $N_{10}$ with intensity rates of $\la_{01}$ and $\la_{10}$ respectively. The process $N_{i,1-i}$, $i\in \{0,1\}$ can be viewed as a counting process for the number of transitions of $M$ between state $i$ and $1-i$. Thus, $M$ has the representation $dM_t = 1_{\{M_{t-}=0\}} dN_{01}(t) - 1_{\{M_{t-}=1\}}dN_{10}(t)$.

The process $M$ is a proxy for the market liquidity.
When $M_t = 0$, the market is liquid and the usual Black-Scholes model for the risky asset (henceforth called `stock') price applies:
\begin{align}
dS_t = \mu S_t \,dt + \sigma S_t \, dW_t,
\end{align}
with $(W_t)$ a one-dimensional Wiener process on a stochastic basis $(\Omega, \mathcal{F}, \mathbb{P})$.
When $M_t = 1$, the market is illiquid; the stock is non-traded and its price experiences a deterministic exponential growth at rate $\alpha \le r$:
\begin{align*}
dS_t = \alpha S_t \, dt.
\end{align*}
Moreover, at the times of transition from $M_{t-} = 0$ to $M_t = 1$, the stock price experiences an instantaneous drop of $L\%$, $S_t = (1-L)S_{t-}$. This jump represents the abrupt decrease in stock price due to a catastrophic market event (such as 9/11 in the USA, collapse of Lehman Brothers, etc.). Overall, the stock obeys a hybrid jump-diffusion model
\begin{align}\label{eq:s-dynamics}
dS_t = \vec{\mu}_{M_t} S_t \,dt + \vec{\sigma}_{M_t} S_t \, dW(t) - L S_t 1_{\{M_{t-} = 0\}} dN_{01}(t),
\end{align}
with $\vec{\mu} = [\mu, \alpha]^T, \; \vec{\sigma} = [\sigma, 0]^T$.

\subsection{Wealth Dynamics}
We consider an economic agent who solves an optimal investment-consumption problem in this regime-switching market. Namely, at any given date $t$  the agent (i) invests a fraction $\pi_t$ of her total wealth in the stock; (ii) invests the remainder of her wealth in a bank account paying interest at fixed rate $r$; (iii) consumes at rate $c_t$ per unit time. We make the additional assumption that \emph{consumption is only financed through cash}. This introduces an additional \emph{cash crunch} constraint whereby the agent may run out of money to consume during an illiquidity shock.

Denote by $(X_t)$ the total wealth of the agent at date $t$. Then $(X_t)$ evolves according to
\begin{align}\notag
\frac{dX_t}{X_t} & = r (1 - \pi_t) \, dt + \pi_t\, \frac{dS_t}{S_t} - \frac{c_t}{X_t} \,dt \\
\label{eq:x-dynamics}
& = (r-\frac{c_t}{X_t}) \,dt + 1_{\{M_{t-} = 0\}} \bigl\{\pi_t (\mu-r) \,dt + \pi_t \sigma \,dW(t) - L \pi_t \,dN_{01}(t)\bigr\} + 1_{\{M_{t-} = 1\}} \pi_t(\alpha-r)  \, dt.
\end{align}
The instantaneous jump of size $L$ in the asset price at the beginning of a liquidity shock has a double effect: it decreases overall wealth, and also increases the fraction of money held in cash, mitigating the cash crunch. Overall, if today's wealth is $X_{t-}$ with a proportion $\pi_{t-}$ held in stock, an instantaneous downward jump of $L\%$ in the stock price causes the total wealth to decrease to $X_t = ((1-\pi_{t-} L)X_{t-}$ and the fraction invested in $S$ to shrink to $\pi_t = g(\pi_{t-})$ where
$$g(\pi) \triangleq \frac{\pi(1- L)}{(1 - \pi L)}.$$
Note that $g(\cdot)$ is an increasing convex function from $[0,1]$ onto $[0,1]$ which reduces to the identity map $g(\pi) = \pi$ when $L=0$.

\begin{rem}
Our model can be straightforwardly extended to the case where the jumps of $S$ at transition times of $M$ have an arbitrary discrete distribution $f_L(\cdot)$ which is bounded away from 1 (and such that its realizations are independent of the rest of the model). All the results below would continue to hold in that case after replacing with $g(\pi) = \pi (1-E[L]) \left( 1-\pi E[L]  \right)^{-1}$. We maintain constant jump sizes for ease of notation and interpretation.
\end{rem}

In contra-distinction from the classical Merton setting, the number of shares $\pi_t X_t/S_t$ held by the investor has liquidity constraints.
Namely, admissible trading strategies are required to maintain the above quantity constant on the time intervals where $M_t = 1$. This represents  market illiquidity whereby the agent is completely unable to trade stock, and corresponds to \emph{exogenously} given dynamics of $\pi_t$ when $M_t = 1$. To make sure the agent does not go bankrupt, the constraint $\pi_t \in [0,1]$ is imposed. Thus, no short-sales or leveraged positions are allowed. Conditional on $M_t = 1$, the dynamics of $(\pi_t)$  are (cf.\ \cite[Lemma 2.1]{Diesinger})
\begin{align}\label{eq:pi-dynamics}
d\pi_t = \pi_t[(1-\pi_t)(\alpha-r) + c_t/X_t] \,dt,
\end{align}
where the second term represents the increase in proportion of wealth in stock due to consumption that is paid for with cash.

Our aim is to quantify the effect of market illiquidity on optimal consumption/investment.
%
In particular, two asymptotic regimes are of interest. First, the case $\la_{01}$ being small represents a marginal possibility of liquidity breakdown. Observe that with $\la_{01}=0$ and $M_0 = 0$ we recover the classical Merton problem. Therefore, the asymptotics of the limit $\la_{01}\to 0$ can be viewed as a \emph{regular perturbation} of this model that takes into account the agent's response to illiquidity.

Second, a \emph{homogenized} limiting model is obtained by re-scaling the generator $Q$ of $M$ to be $Q^\eps \triangleq \frac{Q}{\eps}$ for $\eps \ll 1$. This represents a market that exhibits a fast scale of liquidity regimes that impose additional constraints on the trading agent. In other words, both $\la_{01}$ and $\la_{10}$ are large, with a fixed ratio
$\bar{\la} \triangleq \frac{\la_{10}}{\la_{01} + \la_{10}}$ representing the proportion of time spent in liquid regime. We also rescale the jumps as $L(\eps) = \bar{L} \eps$ which in the limit becomes a negative drift on the stock. In the limit $\eps \to 0$ we obtain a \emph{singular perturbation} of the original Black-Scholes model. Such singularly perturbed hybrid diffusions with fast switching were studied by \cite{IlinYinJOTA99}.

\begin{rem}
In the classical Merton problem with strictly positive consumption, the only constraint on trading is no short-sales, $\pi_t \ge 0$. With liquidity freezes, we have the additional no-leveraging constraint $\pi_t \le 1$. Thus, even as $\la_{01} \to 0$ or $\eps \to 0$ above, if $\hat{\pi} > 1$ our value functions do not converge to the original Merton solutions. Consequently, in markets with large Sharpe ratios, a liquidity \emph{gap} would be present even with negligible chance of market freeze. See Section \ref{sec:large-sharpe-log} and Remark \ref{rem:large-sharpe-hara} below.
\end{rem}

\subsection{Consumption on Infinite Horizon}
As our main model, we consider an agent maximizing cumulative utility of consumption on infinite horizon. We assume that consumption is absolutely continuous, so that we can define a consumption rate $c_t$ with respect to the Lebesgue measure on $\R_+$. The objective function is then given in terms of time-additive utility of consumption,
\begin{align}\label{eq:obj}
\sup_{(\pi,c) \in \mathcal{A}} \E^x \left[ \int_0^\infty \e^{-\rho t} u(c_t) \, dt \right],
\end{align}
where $\rho > 0$ is the inter-temporal substitution factor for consumption and $u : \mathbb{R}_+ \to \R$ is the strictly concave utility function satisfying the Inada conditions $\lim_{x \downarrow 0} u'(x) = +\infty$, $\lim_{x \to +\infty} u'(x) = 0$. In \eqref{eq:obj}, the expectation is with respect to measure $\PP^{x}$ conditional on initial endowment $X_0 = x$. Below we will focus on the case of Hyperbolic Absolute Risk Aversion (HARA, also known as CRRA or constant relative risk aversion) utilities, $u(x) = \frac{x^\gamma}{\gamma}$ with $\gamma \in (-\infty,0) \cup (0,1)$ and $u(x) = \log x$ for $\gamma = 0$. With HARA utilities, investment proportions are known to be independent of current wealth which also carries over to our model and simplifies the interpretation of our results.

Admissible strategies for \eqref{eq:obj} are defined as follows. The observable filtration is always $\mathbb{F} = (\F_t)$ generated jointly by $S$ and $M$. The investment process $(\pi_t)$ is required to be an $\F$-adapted process with $\E[ \int_0^t \pi_s^2 \,ds] < \infty$ for any $t < \infty$. Next, we also require that $(c_t)$ be $\F$-adapted with $\E[ \int_0^\infty \e^{-\rho t} |u(c_t)| \, dt] < \infty$. In the case where $u(0) = -\infty$, ruin is infinitely costly and therefore any admissible strategy must additionally satisfy $\PP( \pi_t < 1) = \PP( X_t > 0) = 1$ for any $t<\infty$. We denote by $\mathcal{A}_0$ the set of all admissible trading strategies that satisfy the above constraints, and finally by $\mathcal{A} \subseteq \mathcal{A}_0$ the subset given in \eqref{eq:obj} that in addition satisfies the freeze constraint during liquidity shocks expressed in \eqref{eq:pi-dynamics}.

To solve \eqref{eq:obj}, we employ the usual Dynamic Programming paradigm. Denote by $V^0(x)$ the value function of the agent with initial wealth endowment of $x$ and conditional on $M_0 = 0$; similarly denote by $V^1(\pi,x)$ the value function of an agent with an endowment of $x$ and a fraction $\pi$ of her wealth in stock, conditional on $M_0=1$. Standard stochastic control arguments suggest the following Hamilton-Jacobi-Bellman equation
\begin{align}\label{eq:hjb-inf-horizon} \left\{
\begin{aligned}
\sup_{c \ge 0} & \Bigl\{ -\rho V^1 + (x(r + (\alpha-r) \pi)-c) V^1_x + \pi(1-\pi)(\alpha-r) V^1_\pi + \pi c/x V^1_\pi \\ & \qquad + \lambda_{10}( V^0(x) - V^1(\pi,x)) + u(c) \Bigr\} = 0, \\
\sup_{\pi \in [0,1],c \ge 0} & \Bigl\{ -\rho V^0 + (x(r+(\mu-r) \pi) - c)V^0_x + \half x^2 \pi^2 \sigma^2 V^0_{xx} \\ & \qquad +\la_{01}( V^1( \frac{\pi(1-L)}{1-\pi L}, (1 - \pi L)x) - V^0(x)) + u(c) \Bigr\} = 0, \\
\lim_{x \downarrow 0} V^1(\pi,x) & = u(0),\quad\text{and}
\qquad \lim_{x \downarrow 0} V^0(x) = u(0),
\end{aligned}\right.
\end{align}
In the illiquid regime 1 only consumption can be optimized; the liquidity shock expires at rate $\la_{10}$ and once it is over liquid trading resumes and remaining utility will then be $V^0(x)$. In the meantime, the fraction invested in stock evolves according to \eqref{eq:pi-dynamics} and total wealth $X_t$ according to \eqref{eq:x-dynamics}. Analogously, starting in the liquid regime the agent optimizes the positive fraction of wealth invested in the stock and consumption. All the while she is anticipating the possibility of liquidity shocks that occur at rate $\la_{01}$ and entail immediate loss of $L\%$ in the stock value. The different dynamics of $(X_t)$ under the two liquidity regimes lead to modified  derivative terms with respect to $x$ in \eqref{eq:hjb-inf-horizon}.

We now state the following verification theorem; the proof is given in the Appendix.
\begin{prop}\label{thm:verify}
Suppose that $J^0(x) \in \mathcal{C}^2(\R_+), J^1(\pi,x) \in \mathcal{C}^{1,1}(\R_+ \times [0,1))$
are two smooth functions with polynomial growth in $x$ that satisfy the HJB equation \eqref{eq:hjb-inf-horizon}.
Then $J^0(x) \ge V^0(x)$, $J^1(\pi,x) \ge V^1(\pi,x)$. Moreover, if there exist continuous functions $\pi^*, c^{0,*}, c^{1,*}$ satisfying
\begin{align*}
\pi^*(x) &\in \argmax_{\pi \in [0,1]} \Big\{ x(\mu-r)\pi J^0_x + \half x^2 \pi^2 \sigma^2 J_{xx}^0  + \la_{01} J^1( g(\pi), (1 - \pi L)x) \Big\},\\
c^{0,*}(x) &\in \argmax_{c \ge 0} \{ - c J^0_x + u(c) \}, \qquad c^{1,*}(\pi,x) \in \argmax_{c\ge 0} \{ - c J^1_x + \pi c/x J^1_\pi + u(c) \},
\end{align*}
such that the strategy $(\pi,c)$ is admissible,
then $J^0(x) = V^0(x)$ and $J^1(\pi,x) = V^1(\pi,x)$ and $\pi^*, c^{0,*}, c^{1,*}$ are optimal strategies.
\end{prop}

\subsection{Asymptotics for Small $\la_{01}$}
For the asymptotic case $\la_{01}\to 0$, we may uncouple the two value functions above. Indeed, on finite horizon when $\la_{01}$ is sufficiently small, the probability of more than one $0 \to 1$ transition of $(M_t)$ is negligible.
In fact, the probability of $n$ such transitions is approximately
(using $\la_{01} \ll \la_{10}$) $\mathcal{O}( \e^{-n \la_{01}}) =
o(\e^{-\la_{01}})$ for $n>1$. Thus, assuming that $M_0 = 0$, we need to consider at most two liquidity regime-changes in $(M_t)$. For this purpose we will define the value functions $V$ and $W$ as follows: $V(x)$ represents the maximum value an agent can extract when $M_0 = 0$ and no liquidity shock has occurred so far; $W(\pi,x)$ represents the maximum value when $M_0 = 1$ so the liquidity shock is already present and after its conclusion we revert to the classical Merton setting (with value function $\hat{V}(x)$).

\begin{lemma}
We have $V^0(x) = \hat{V}(x) + \la_{01} V(x) + \mathcal{O}(\la_{01}^2)$  where
\begin{align}\label{eq:uncoupled-hjb}\left\{
\begin{aligned}
\sup_{\pi \in [0,1],c \ge 0} & \Bigl\{ -\rho V + (x(r+(\mu-r) \pi) - c)V_x + \half x^2 \pi^2 \sigma^2 V_{xx} \\ & \qquad +\la_{01}( W(\frac{\pi(1-L)}{1-\pi L}, (1 - \pi L)x) - \hat{V}(x)) + u(c) \Bigr\} = 0;\\
\sup_{c \ge 0} & \Bigl\{ -\rho W + (x(r + (\alpha-r) \pi)-c) W_x + \pi(1-\pi)(\alpha-r) W_\pi + \pi c/x W_\pi \\ & \qquad + \lambda_{10}( \hat{V}(x) - W(\pi,x)) + u(c) \Bigr\} = 0.
\end{aligned}\right.
\end{align}
\end{lemma}
\begin{proof}
Consider the formal asymptotic expansions $V^0(x) = V^{00}(x) + \la_{01}V^{01}(x) + \la_{01}^2 V^{02}(x) + \mathcal{O}(\la_{01}^3)$ and $V^1(\pi,x)  = V^{10}(\pi, x) + \la_{01}V^{11}(\pi, x) + \la_{01}^2 V^{12}(\pi, x) + \mathcal{O}(\la_{01}^3)$. Plugging into \eqref{eq:hjb-inf-horizon} and matching powers of $\la_{01}$ shows that
$V^0(x) = \hat{V}(x)$ is the solution of the classical Merton problem, while $V^{10}(\pi,x) \equiv W(\pi,x)$ solves the second equation in \eqref{eq:uncoupled-hjb}. Matching terms of $\mathcal{O}(\la_{01})$ in the equation for $V^0$ shows that $V^{01}(x) \equiv V(x)$ solves the first equation in \eqref{eq:uncoupled-hjb}.
\end{proof}

\section{Log Utility}\label{sec:log-utility}
We now specialize to the case $u(x) = \log x$. We first recall the following classical result  (see \cite[Chapter 4]{KarShreveMethods}).

\begin{prop}[Example 9.24, p. 150 in \cite{KarShreveMethods}]\label{lem:merton-inf-log}
The solution of the optimal consumption problem with log-utility is given by $\hat{V}(x) = \frac{1}{\rho} \log x + \hat{h}$ with
\begin{align}\label{eq:merton-inf-horizon-log}
\hat{h} =  \frac{r/\rho - 1 + \log \rho}{\rho} + \frac{\theta^2}{2 \rho^2}, \qquad \theta \triangleq \frac{(\mu-r)}{\sigma}.
\end{align}
Moreover, the optimal consumption/investment strategies are given by $\hat{c} = \rho x$, $\hat{\pi} = \frac{\mu-r}{\sigma^2} = \frac{\theta}{\sigma}$.
\end{prop}

In the special case $\alpha=r$, meaning that the excess return of the stock is zero in the illiquid regime, an explicit solution to \eqref{eq:hjb-inf-horizon} is possible.
\begin{thm}\label{thm:log}
The solution of \eqref{eq:hjb-inf-horizon} for $u(x) = \log x$ and $\alpha=r$ is given by $V^1(\pi,x) = \frac{1}{\rho} \log x + h(\pi)$ and $V^0(x) = \frac{1}{\rho} \log x + b$ where
\begin{equation}\label{eq:h-log}
h(\pi)=\frac{
r/\rho-1 + \log \rho + \la_{10} b + \log\left(1-\pi^{(1+\lambda_{10}/\rho)}\right)}{\rho+\lambda_{10}},
\end{equation}
and
\begin{equation}\label{eq:b-log}
b = \frac{r/\rho-1 + \log \rho}{\rho} + \zeta(\pi^*) \cdot \frac{ \rho + \la_{10}}{\rho(\rho+ \la_{01}+\la_{10})},
\end{equation}
\begin{align}\label{eq:zeta}
\zeta(\pi^*) \triangleq \sup_{\pi\in[0,1]} \left\{ \frac{(\mu-r) \pi - \pi^2 \sigma^2/2 }{\rho} +  \frac{\la_{01}}{\rho+\lambda_{10}}\log\left(1-g(\pi)^{(1+\lambda_{10}/\rho)}\right)  + \frac{\la_{01}}{\rho} \log (1-\pi L)  \right\}.
\end{align}
The optimal investment strategy in the liquid regime is the maximizer $\pi^*$ in \eqref{eq:zeta} and the optimal consumption schedules are
\begin{align}\label{eq:log-c1-star}
 c^{1,*}(\pi,x) &= \rho x \cdot (1-\pi^{(1+\lambda_{10}/\rho)});\\ \label{eq:log-c0-star}
 c^{0,*}(x) & = \rho x.
\end{align}

\end{thm}

Theorem \ref{thm:log} presents a complete solution of the two-state liquidity model for log-utility. The two asymptotic regimes will be studied in the next subsection; in the meantime let us make a brief remark about \eqref{eq:log-c1-star}-\eqref{eq:log-c0-star}. First, observe that optimal consumption in the liquid regime $c^{0,*}$ is \emph{myopic} in the sense of being unaffected by liquidity shocks. In the illiquid regime, $c^{1,*}(\cdot,x) \le c^{0,*}(x)$ is reduced due to the cash crunch. This curtailment of consumption increases as $\pi_t \uparrow 1$ and the agent becomes more liquidity-constrained. The effect decreases as the duration of liquidity shocks shrinks ($\la_{10} \uparrow$) or intertemporal discounting decreases ($\rho \downarrow$).

Theorem \ref{thm:log} also explains the behavior of stock investments in our model. When $M_t = 0$, the agent maintains a constant fraction of her wealth in stock, $\pi^*$ given in \eqref{eq:zeta}. When a liquidity shock occurs at date $\tau$, the illiquid fraction jumps to $\pi_{\tau} = g(\pi^*)$ and then starts growing according to $\pi_t = \Pi(t-\tau)$ where $\Pi$ is the solution of the deterministic differential equation $d\Pi(t) = \rho \Pi(t)(1-\Pi(t)^{1+\la_{10}/\rho}) \,dt$, $\Pi(0) = g(\pi^*)$ resulting from combining \eqref{eq:pi-dynamics} and \eqref{eq:log-c1-star}. Once the liquidity shock expires, wealth is immediately rebalanced to $\pi_t = \pi^*$ again.

\begin{proof}
Comparing with the form of the Merton value function in Proposition \ref{lem:merton-inf-log}, we make the ansatz given in the statement of the Theorem, $V^1(\pi,x) = \frac{1}{\rho} \log x + h(\pi)$ and $V^0(x) = \frac{1}{\rho} \log x + b$.

Plugging-in into \eqref{eq:hjb-inf-horizon} we find that $h$ must solve
the nonlinear first order ode
\begin{align}\label{eq:h-ode}
-(\rho + \la_{10})h + r/\rho - 1 + \la_{10}b - \log(\rho^{-1} - \pi h'(\pi)) = 0
\end{align}
with $b$ the constant in $V^0$ and corresponding optimal consumption level of
\begin{align*}
{c}^{1,*}(\pi, x) = \frac{x}{\rho^{-1} - \pi h'(\pi)}. 
\end{align*}

Taking $\pi =0$ removes the differential term and we find
\begin{align}\label{eq:h-zero-value}
h(0) = \frac{
r/\rho-1  + \log \rho+\la_{10} b }{\rho + \la_{10}}.
\end{align}
However, since zero is a singular point for \eqref{eq:h-ode} with $h'(0) = 0$, the value of $h(0)$ is not sufficient to specify uniquely the solution of \eqref{eq:h-ode}. Indeed, the general solution of \eqref{eq:h-ode} is
\begin{equation}\label{eq:gen-log-solution}
h(\pi)=h(0)+\frac{1}{\rho+\lambda_{10}}\log\left(1- C \pi^{(1+\lambda_{10}/\rho)}\right),
\end{equation}
with integration constant $0 < C < 1$. To find $C$, we invoke the cash-crunch constraint which implies $\lim_{\pi \uparrow 1} V^1(\pi,x) = - \infty$ or $\lim_{\pi \uparrow 1} h(\pi) = \lim_{\pi \uparrow 1} h'(\pi)=-\infty$. Subject to this restriction we obtain $C=1$ in \eqref{eq:gen-log-solution} which combined with the expression for $h(0)$ in \eqref{eq:h-zero-value} yields \eqref{eq:h-log}. Furthermore, \eqref{eq:h-log} leads to
\begin{equation*}
h'(\pi)=\frac{-\pi^{(\lambda_{10}/\rho)}}{\rho(1-\pi^{(1+\lambda_{10}/\rho)})}
\end{equation*}
and together with the first-order-condition for $c^{1,*}$ implies \eqref{eq:log-c1-star}.

Returning to the liquid regime and substituting the ansatz of $V^0(x)$ we obtain that $b$ must solve
$$
-(\rho + \la_{01})b + (r/\rho -1 + \log \rho) + \la_{01} h(0) + \zeta(\pi^*) =0,
$$
with $\zeta(\cdot)$ defined in \eqref{eq:zeta}.
We note that $\zeta(\pi^*)$ is independent of $b$; on the other hand recall that $h(0)$ in \eqref{eq:h-zero-value} is a function of $b$. Simplifying we end up with \eqref{eq:b-log}.

Since $\zeta(\pi)$ contains all the terms involving $\pi$, its maximizer $\pi^*$ must be the optimal investment strategy. When $\la_{01} = 0$, the maximizer is $\hat{\pi}$ from \eqref{eq:merton-inf-horizon-log}. Consequently, because $h'(\pi) <0$ for all $\pi$, the optimal investment fraction $\pi^* < \hat{\pi}$ is less than the classical Merton fraction in Proposition \ref{lem:merton-inf-log}.

Finally, the ansatz for $V^0$ implies immediately that $c^{0,*} = (V^0_x)^{-1} = \rho x$ which matches the consumption schedule of Proposition \ref{lem:merton-inf-log}.
\end{proof}

When $\alpha \neq r$, the ode \eqref{eq:h-ode} has another $h'(\pi)$ term and is no longer separable. As a result, no closed-form solution is possible; however the rest of the conclusions of Theorem \ref{thm:log} continue to hold, e.g.~\eqref{eq:b-log}, \eqref{eq:log-c0-star}. The condition $\alpha=r$ can be interpreted as zero excess return for the stock in the illiquid regime. By adjusting the value of $L$, the negative jump at the beginning of a liquidity shock, the general $\alpha \neq r$ setup can be effectively approximated.

To interpret Theorem \ref{thm:log}, we compare the liquid value function $V^0(x)$ and the classical Merton solution $\hat{V}(x)$ in  terms of the efficiency loss $\bt$ introduced by \cite{Rogers00}.
Namely, let $0\le \bt<1$ be the percentage of total wealth that needs to be subtracted from an agent facing the classical Merton problem to make her utility equal to that of the agent facing liquidity shocks. Then for log-utility, $\bt\equiv\bt^{(\log)}$ satisfies
$$
\frac{1}{\rho} \log x + b = \frac{1}{\rho} \log \left((1-\bt^{(\log)})x \right) + \hat{h}  \quad \Longleftrightarrow \quad \bt^{(\log)} = 1-\exp( \rho( b-\hat{h}) ).
$$
Simplifying the above expression using \eqref{eq:b-log} we find
\begin{equation}\label{eq:perc-loss-log}
\bt^{(\log)} = 1-\exp \left( \frac{-\theta^2}{2\rho}  + \zeta(\pi^*)  \frac{ \rho + \la_{10}}{\rho + \la_{01}+\la_{10}} \right).
\end{equation}
Thus, the loss in utility depends on the function $\zeta(\cdot)$ in \eqref{eq:zeta}, as well as on the squared Sharpe ratio $\theta^2$ of the stock.

\subsection{Small $\la_{01}$ Asymptotics}
We recall that $\pi^*$ satisfies the first order condition from \eqref{eq:zeta}, $$(\mu -r)- \sigma^2 \pi^* + \la_{01} \left[\rho h'(g(\pi^*))g'(\pi^*)  - \frac{L}{1-\pi^* L} \right]=0.$$

When $\hat{\pi}<1$, we may expand around $\hat{\pi}$ to obtain
an asymptotic expression of the optimal investment fraction
$\pi^*$ for $\lambda_{01}$ small:
\begin{align}\notag
\pi^* &= \hat{\pi} - \lambda_{01} \pi_1 + \O(\la_{01}^2) \quad \text{where}\\ \notag
 \pi_1 &\triangleq \frac{1}{\sigma^2}
\left( \frac{L}{1 - \hat{\pi} L} - \rho h'(g(\hat{\pi})) g'(\hat{\pi}) \right) \\ \label{eq:pi-log-correction}  & =
\frac{1}{\sigma^2}
\left( \frac{L}{1 - \hat{\pi} L} + \frac{1 -L}{(1-\hat{\pi}L)^2} \cdot    \frac{g(\hat{\pi})^{\lambda_{10}/\rho}}{1-g(\hat{\pi})^{(1+\lambda_{10}/\rho)}}\right).
\end{align}
In a typical market $\rho \ll \lambda_{10}$ and therefore $g(\hat{\pi})^{\lambda_{10}/\rho}$ is negligible; it follows that the effect of liquidity shocks on stock investment mainly depends on $L$.

To understand the small $\la_{01}$-asymptotics of the value functions we use \eqref{eq:uncoupled-hjb}. This implies that to leading order, we should replace $b$ with $\hat{h}$ in the expression of $h(\pi)$ in \eqref{eq:h-zero-value} and therefore $b$ satisfies
$$
-\rho b + (r/\rho -1 + \log \rho) + \la_{01} h(0) + \zeta(\pi^*) - \la_{01} \hat{h} + \O(\la_{01}^2)=0.
$$
Plugging-in the expansion of $\pi^*$ in \eqref{eq:pi-log-correction} and simplifying leads to
$ b = \hat{h} + \la_{01}b_1 + \O(\la_{01}^2)$ where
\begin{align*}
b_1 = \frac{1}{\rho} \Bigr[ \frac{
r/\rho-1  + \log \rho+\la_{10} \hat{h} }{\rho + \la_{10}} - \hat{h} + \frac{1}{\rho}\log(1-\hat{\pi}L)
   + \frac{1}{\rho+\la_{10}} \log( 1- g(\hat{\pi})^{1+\la_{10}/\rho}) \Bigr].
\end{align*}
When $\rho \ll \la_{10}$, the above implies that the asymptotic efficiency loss is
\begin{align}\label{eq:perc-loss-small-log}
\bt = \la_{01}\bt_1 + \O(\la_{01}^2)\quad\text{ with }\quad\bt_1 \simeq \frac{\theta^2}{2 \la_{10}} - \frac{\log( 1- \hat{\pi}L)}{\rho} .
\end{align}
Thus, the asymptotic efficiency loss $\bt_1$ increases quadratically in the squared Sharpe ratio of the stock and has a logarithmic relationship with respect to stock jumps during liquidity shocks.

\subsection{Assets with Large Sharpe Ratios}\label{sec:large-sharpe-log}
The above expansion is valid if $\hat{\pi} < 1$, i.e.\ the Sharpe ratio of the risky asset is not too large.
If $\hat{\pi}\ge 1$, then the constraint $\pi^* <1$ remains in force throughout and it is the case that $\lim_{\la_{01}\to 0} \pi^* = 1$. Supposing that $g(\pi^*) = 1-\eps$ for some small parameter $\eps$ and using $1-(1-\eps)^a \simeq a \eps$ for any $a>0$ the first order condition for $\pi^*$ reduces to
$$
\mu -r -\sigma^2 - \frac{\la_{01}}{(1+ \la_{10}/\rho) \eps(1-L)} - \la_{01} \frac{L}{1-L} + \O(\la_{01}) = 0.
$$
Therefore, $1-\pi^*$ is asymptotically proportional to  $\la_{01}$ and
\begin{align}
1-\pi^* & = \la_{01}\frac{\rho}{((\mu-r)-\sigma^2)(\rho+\la_{10})} + \O(\la_{01}^2). \label{eq:log-pi-bigger-1}
\end{align}
Using this expansion inside $V^0(x)$ we find
\begin{align}\label{eq:log-V-pi-bigger-1}
V^0(x) & = \left[\frac{1}{\rho} \log x + \frac{r/\rho-1+\log \rho}{\rho} + \frac{\mu-r-\sigma^2/2}{\rho^2}\right]  + \Bigl\{ \frac{1}{\rho(\rho + \la_{10})}\la_{01}\log \la_{01} \Bigr\} + \O(\la_{01}),
%
%
\end{align}
where the term in the square brackets corresponds to the value of a Merton strategy with $\pi \equiv 1$.
Note that the leading term in the expansion is now of order $\la_{01}\log(\la_{01})$ due to the cash crunch.

\section{HARA Utility of Consumption}\label{sec:power-utility}
We next consider the case of HARA power utilities, $u(x) =\frac{x^{\gamma}}{\gamma}$, where $\gamma \in (-\infty,1) \setminus \{0 \}$ is the risk-aversion coefficient. We start with qualitative analysis of the nonlinear ode satisfied by the value function in the illiquid regime and then derive the small $\la_{01}$ asymptotics for the liquid-regime solution. The remaining subsections deal with special values of $\gamma$, whereby closed-form solutions are possible for $W(\pi,x)$.
Throughout this section we make the standing assumption that asset value during the liquidity shocks grows at the riskless rate $\alpha =r$.

\begin{prop}[Example 9.22, p. 149 in \cite{KarShreveMethods}] \label{lem:merton-inf-hara}
Let $u(x) = x^\gamma/\gamma$, $\gamma \in (-\infty,0) \cup(0,1)$. Define
$$
\delta \triangleq \rho - \gamma r - \frac{ \theta^2 \gamma}{2(1-\gamma)}.
$$
If $\gamma > 0$ and $\delta < 0$ then $\hat{V}(x) = +\infty$ and infinite utility  of consumption can be extracted. Otherwise, the Merton optimal investment/consumption problem has the solution
\begin{align}
\hat{V}(x) = \hat{f} \cdot \frac{  x^\gamma}{\gamma}, \qquad\text{where}\quad \hat{f} \triangleq  \left(\frac{1-\gamma}{\delta}\right)^{1-\gamma}.
\end{align}
The optimal strategies are
\begin{align}\label{eq:merton-pi-power}
\hat{\pi} = \frac{1}{1-\gamma}\frac{(\mu-r)}{\sigma^2}, \qquad \hat{c}(x) = \frac{\delta}{1-\gamma} x.
\end{align}
\end{prop}

Below we will work with the power-type utility functions that admit a convenient scaling property in wealth. Namely, in analogue to the form of the classical Merton solution, we have
\begin{lemma}\label{lem:hom-degree-gamma}
For $\gamma \in (-\infty,0) \cup (0,1)$, the value functions $V^i$ are homothetic of degree $\gamma$ in $x$.
\end{lemma}

\begin{proof}
This standard result follows from the linearity of the wealth dynamics \eqref{eq:x-dynamics} which imply that
if $(c^x_t)$ is an admissible consumption schedule for the problem with initial endowment $x$, then $(\beta c^x_t)$ is admissible for initial endowment of $\beta x$, $\beta > 0$. Moreover, the admissibility set $\mathcal{A}$ of trading strategies is independent of the endowment. Suppose now that $(c^{x,*}_t)$ is an $\eps$-optimal consumption control for $V^0(x)$. Then using above,
\begin{align*}
V^0(\beta x) \ge \E \left[ \int_0^\infty \e^{-\rho s} (\beta c^{x,*}_s)^\gamma \,ds \right] \ge \beta^\gamma V^0(x) -\eps.
\end{align*}
Applying the same inequality with $(c^{\beta x, *}_t)$ gives that $V^0(x) \ge \beta^{-\gamma} V^0(\gamma x) -\eps$ and since $\eps$ is arbitrary, we must have $V^0(\beta x) = \beta^\gamma V^0(x)$. Similar arguments work for $V^1(\pi,x)$.
\end{proof}

We first concentrate on the case of small $\la_{01}$ and work with the uncoupled equations \eqref{eq:uncoupled-hjb}.
Based on Lemma \ref{lem:hom-degree-gamma}, we make the guess $W(\pi, x) = f(\pi)x^\gamma/\gamma$. Simplifying shows that $f$ solves
\begin{align}\label{eq:f-power}
(-\rho + \gamma r) f + (1-\gamma)[f -\pi \gamma^{-1} f']^{\frac{\gamma}{\gamma-1}} + \la_{10}[\hat{f} - f] = 0.
\end{align}
The boundary conditions for \eqref{eq:f-power} will be discussed shortly. For $\pi=0$, we can directly find $f(0)$ as the solution of $f(0) = H^{-1}(0)$ where
\begin{align}\label{eq:f0-bound-cond}
H(x) \triangleq (-\rho + \gamma r - \la_{10})x + (1-\gamma)x^{\frac{\gamma}{\gamma-1}} + \la_{10}\hat{f}.
\end{align}
Making the substitution $\e^{-z} = \pi$, $\phi(z) \triangleq (\rho-\gamma r + \la_{10})f(\e^{-z}) - \la_{10}\hat{f}$ we transform \eqref{eq:f-power} into
\begin{equation}
\frac{1}{\gamma}\phi'(z)= (\rho - \gamma
r + \la_{10}) \left(\frac{\phi(z)}{1-\gamma}\right)^{1-1/\gamma}-\!\phi(z)-
\la_{10} \hat{f}, \qquad z \ge 0. \label{eq:h-deriv}
\end{equation}
The last equation is a nonlinear separable first-order ode on $\R_+$ with the right-hand-side resembling the Bernoulli ode except for the constant term $\la_{10}\hat{f}$. Its special form allows for much tractability.

Consider first the case $\gamma < 0$. It can then be checked that $H(x)$ in \eqref{eq:f0-bound-cond} is increasing on $[\hat{f}, \infty)$ and $H(\hat{f}) < 0$ so that \eqref{eq:f0-bound-cond} has a unique root with $f(0)> \hat{f}$. Transferring into \eqref{eq:h-deriv} we obtain
$$
\lim_{z \to \infty} \phi(z) = (\rho - \gamma r + \la_{10})f(0) -\la_{10}\hat{f} = (1-\gamma)f(0)^{\frac{\gamma}{\gamma-1}},
 $$
and the ode \eqref{eq:h-deriv} has a horizontal asymptote. As for log-utility, the cash constraint implies that $V^1(1,x) = -\infty$ which means that $f(1) = \phi(0) = +\infty$. Moreover, we have $\phi'(0) = -\infty$ and $\phi(z)$ is monotone decreasing on $[0,\infty)$.

Next, consider $\gamma > 0$. When hyperbolic risk aversion is positive we have $u(0) = 0$, and the agent can tolerate zero consumption. Therefore, $V^1(1,x)$ and $W(1,x)$ are finite. Taking $c=0$ directly in the pde \eqref{eq:uncoupled-hjb} for $W(\pi,x)$ and making the ansatz $W(1,x) = f^+(1)\frac{x^\gamma}{\gamma}$ we obtain
\begin{align*}
f^+(1) = \frac{\la_{10}\hat{f}}{\rho - \gamma r + \la_{10}} < \hat{f} < f(0).
\end{align*}
Recall that $\rho - \gamma r > 0$ is a necessary condition to guarantee finite $\hat{f}$ and
holds thanks to assumptions of Proposition \ref{thm:verify}.
Note that since the marginal utility of consumption at zero is infinite $u'(c)\big|_{c=0} = +\infty$, we still have $\phi'(0) = f'(1) = +\infty$. Conversely, the function $H(\cdot)$ in \eqref{eq:f0-bound-cond} is now decreasing and has a unique root $f(0) > \hat{f}$. To summarize, when $\gamma>0$, $\phi(\cdot)$ is finite, with the singular boundary condition $\phi(0) = 0, \phi'(0) = +\infty$, and is monotonically increasing on $[0,\infty)$ to the horizontal asymptote $(1-\gamma)f(0)^{\frac{\gamma}{\gamma-1}}$.

The following technical lemma, proved in the Appendix, clarifies the structure of \eqref{eq:h-deriv}.
\begin{lemma}\label{lem:power-ode-soln}
The ordinary differential equation \eqref{eq:h-deriv} has a unique solution satisfying the above boundary conditions. Moreover, for $\gamma<0$ we have that $\phi(z)$ (resp.\ $f(\pi)$) is of polynomial growth as $z \downarrow 0$ (resp.\ $\pi \to 1$).
\end{lemma}

Recall that the optimal consumption level is given in terms of $\phi$ as
$$
c^{1,*}(\pi,x) = \left(\frac{\phi( -\log(\pi) )}{1-\gamma}\right)^{1/\gamma} \cdot x.
$$
By Lemma \ref{lem:power-ode-soln}
the consumption level $c^{1,*}(\pi)$ is therefore polynomial in $(1-\pi)$. Let $\pi^{1,*}_t$, with initial condition $\pi^{1,*}_0 = \pi^{0,*}$, be the illiquid fraction of wealth of stocks when consuming according to $c^{1,*}$ and remaining in the illiquid regime for $t$ time units. Then $d\pi^{1,*}_t = \pi^{1,*}_t c^{1,*}(\pi^{1,*}_t,X_t)/X_t \,dt \le (1-\pi_t)^{C_1} \,dt$ for some $C_1$. It follows that
there exists a constant $C_2$, such that $\pi^{1,*}_t \le 1 - t^{-C_2}$ and the cash crunch is approached at a polynomial speed.

To summarize, solutions of \eqref{eq:f-power} can be written in terms of indefinite integrals. In particular, when $-\gamma^{-1}$ is a positive integer (i.e.\ $\gamma = -1, -1/2, -1/3$), the ode in \eqref{eq:h-deriv} is of the type known as Chini differential equation \citep{KamkeBook} and may be solvable analytically. Below, we will explore further the cases $\gamma=-1$ (leading to a quadratic ode of hyperbolic type), $\gamma=-1/2$ (leading to a variant of the Abel cubic ode) and $\gamma=0.5$ (leading to a reciprocal ode again of hyperbolic type).

Returning to the liquid regime, let us denote $\xi(\pi) \triangleq f \left( \frac{\pi(1-L)}{1-\pi L} \right) \cdot (1-\pi L)^\gamma$. We make the ansatz $V(x) = \frac{x^\gamma}{\gamma} \cdot b$ which leads to the following algebraic equation for the constant $b$:
\begin{align}\label{eq:power-v-func}
\frac{-\rho + \gamma r - \la_{01}}{\gamma} b + \frac{1-\gamma}{\gamma} b^{\frac{\gamma}{\gamma-1}} + \sup_{\pi \in [0,1]} \left\{ (\mu-r)\pi b + \frac{\gamma-1}{2} \sigma^2 \pi^2 b + \la_{01}\frac{\xi(\pi)}{\gamma} \right\} = 0.
\end{align}

\subsection{Small $\la_{01}$ Asymptotics}
 Expanding $b$ and $\pi^*$ simultaneously in powers of $\la_{01}$ as
$$
 b = \hat{f} + \la_{01}b_1 + \O(\la_{01}^2) \qquad \text{and}\quad \pi^* = \hat{\pi} + \la_{01}\pi_1 + \O(\la_{01}^2),
$$
we find that for $\la_{01}$ small the first-order correction to $\pi^*$ can be written as
\begin{align}\label{eq:pi-power-func}
\pi_1 = -\frac{(\mu-r)b_1 - \xi'(\hat{\pi})}{\gamma (1-\gamma)\sigma^2 \hat{f}} + \O(\la_{01}^2).
\end{align}
Plugging back into \eqref{eq:power-v-func}, and simplifying we find
\begin{align*}
-\delta (\hat{f} + \la_{01}b_1) +
(1-\gamma) (\hat{f} + \la_{01}b_1)^{\frac{\gamma}{\gamma-1}} + \la_{01}(\xi(\hat{\pi}) -\hat{f} - \la_{01}b_1)+ \O(\la_{01}^2) = 0,
\end{align*}
so that 
\begin{align}
b_1 & = \frac{\xi(\hat{\pi}) - \hat{f}}{\delta + \gamma \hat{f}^{\frac{1}{\gamma-1}} }  = \frac{(\xi(\hat{\pi}) - \hat{f})(1-\gamma)}{\delta} \label{eq:power-small-b}.
\end{align}

The derived formula \eqref{eq:power-small-b} has several implications. First, translating in terms of efficiency loss we have $\bt^{(\gamma)} = \la_{01}\bt^{(\gamma)}_1+\O(\la_{01}^2)$ where
\begin{align}\label{eq:perc-loss-power}
\bt^{(\gamma)}_1 = \frac{1}{\la_{01}}\left(1-\left( \frac{\hat{f}}{\hat{f} + \la_{01} b_1} \right)^{1/\gamma} \right) = -\frac{b_1}{\gamma \hat{f}}  = \left(\frac{ \xi(\hat{\pi})}{\hat{f}} -1 \right) \frac{(\gamma-1)}{\gamma \delta}.
\end{align}
Again, unless $\hat{\pi}$ is large (e.g.\ bigger than $70\%$), $\xi(\hat{\pi})/\hat{f}$ is close to 1 and the correction due to illiquidity is small. Also, the consumption during the liquid regime is driven by the same expressions,
$$
c^{0,*}(x) = x \cdot \left( \hat{f} + \la_{01} b_1 \right)^{1/(\gamma-1)} = \hat{c}(x) \left( 1 + \la_{01} \frac{b_1}{\hat{f}} \right)^{1/(\gamma-1)}.
$$
Equation \eqref{eq:perc-loss-power} shows that the liquidity efficiency loss is driven by the key function $f(\pi)$ from \eqref{eq:f-power}. In the remainder of this section we investigate special values of $\gamma$ whereby closed-form expressions for $f(\pi)$ (or $\phi(z)$) are possible.

\begin{rem}
The coupled equations \eqref{eq:hjb-inf-horizon} can be treated similarly, by making the ansatz $V^0(x) = \bar{b} x^\gamma/\gamma$ and $V^1(\pi,x) = \bar{f}(\pi) x^\gamma/\gamma$. This leads to the coupled equations for $\bar{f}(\pi)$ and $\bar{b}$, analogously to \eqref{eq:f-power} and  \eqref{eq:power-v-func}.
\end{rem}

\subsection{Hyperbolic Utility}\label{sec:hyperbolic-utility}
For the case of hyperbolic utility, $\gamma=-1$, the nonlinear ode for the illiquid regime
\eqref{eq:f-power} can be solved analytically.

\begin{thm}\label{thm:hyperbolic}
Let $\beta \triangleq (\rho+r+\la_{10})$.
The solution of \eqref{eq:hjb-inf-horizon} with $u(x) = -\frac{1}{x}$ is $V^0(x) = -B/x$, $V^1(\pi,x) = -F(\pi)/x$ where
\begin{align}\label{eq:f-hyperbolic}
F(\pi) & =\frac{1}{\beta^2}\Bigl\{1 + \eta(-1)^2 +2
\eta(-1) \cdot \frac{1+\pi^{\eta(-1)}}{1-\pi^{\eta(-1)}}\Bigr\}, \\
\label{eq:coupled-hyberbolic}
(\rho+r+\la_{01})B  & - 2\sqrt{B} + \sup_{\pi \in [0,1]} \Big\{ (\mu-r)\pi B - \sigma^2 \pi^2 B - \frac{\la_{01}}{1-\pi L} F( g(\pi)) \Big\} = 0, \\ \notag
\text{and}\quad \eta(-1) & = \sqrt{1+\beta \la_{10}B}.
\end{align}
The optimal consumption schedules are $c^{0,*}  = x/B^2$ and
\begin{align*}
c^{1,*}(\pi,x) & = x \frac{ \beta (1-\pi^{\eta(-1)})}{\sqrt{ (1+\pi^{\eta(-1)})[ (1+\eta(-1))^2 + \pi^{\eta(-1)} (1-\eta(-1))^2 ]}} 
.
\end{align*}
\end{thm}

The proof of Theorem \ref{thm:hyperbolic} follows straightforwardly by checking that the solution $F(\pi)$ satisfies the corresponding version of \eqref{eq:f-power} and then substituting back into the liquid regime equation \eqref{eq:power-v-func}.
The algebraic equation \eqref{eq:coupled-hyberbolic} can be solved easily numerically for $B$. For small $\la_{01}$ the solution of the uncoupled HJB equation \eqref{eq:uncoupled-hjb} is obtained by using $\eta(-1) = (1+ \beta \la_{10}\hat{f})^{1/2}$ in \eqref{eq:f-hyperbolic} and keeping \eqref{eq:coupled-hyberbolic} as is.

\begin{rem}\label{rem:large-sharpe-hara}
In the case of the large Sharpe ratio $\hat{\pi}>1$ we can use \eqref{eq:f-hyperbolic} to re-compute the small $\la_{01}$-asymptotics. Omitting the details that are very similar to the derivations of \eqref{eq:log-pi-bigger-1}-\eqref{eq:log-V-pi-bigger-1} we find
\begin{align*}
 \pi^* & = 1 - \sqrt{\la_{01}} \pi_1  + o(\la_{01}), \qquad \pi_1 = 2 \beta^{-1} \left\{ \hat{f}(\mu-r-2 \sigma^2)\right\}^{-1/2}.
%
\end{align*}
Note that the correction terms for both $\pi^*$ and $b$ are now of order $\sqrt{\la_{01}}$.
\end{rem}

\subsection{Square-Root Utility}
Another closed-form solution for $W(\pi,x)$ is possible when $\gamma=0.5$ and $u(x) = 2\sqrt{x}$. This case is instructive for showing the interaction of positive risk-aversion with the liquidity freezes.

With $\gamma=0.5$, the ode \eqref{eq:h-deriv} can be re-written in the integral form as
\begin{align*}
 \int dz + C = \int \frac{-2 \phi}{ \phi^2 + \la_{10}\hat{f} \phi - \frac{\rho - 0.5 r + \la_{10}}{2}} d\phi, \qquad \phi(0) = 0.
\end{align*}
The integral on the right hand side can be computed using partial fraction decomposition. The resulting implicit solution is
\begin{equation}\label{eq:g-square-root}
\begin{split}
\e^{z} & = \Bigl( 1+ \frac{ 2 \phi(z)}{\eta(0.5) + \la_{10}\hat{f}} \Bigr)^{-1-(\la_{10}\hat{f})/\eta(0.5)} \cdot \Bigl( 1 - \frac{2 \phi(z)}{\eta(0.5) - \la_{10}\hat{f}} \Bigr)^{-1+(\la_{10}\hat{f})/\eta(0.5)}\\
\end{split}
\end{equation}
where $\eta(0.5)  \triangleq \sqrt{2(\rho-0.5r+\la_{10})+(\la_{10}\hat{f})^2}$.
We remark that for typical parameter values, $\lim_{z \to\infty} \phi(z)$ is small, so that $\phi(z)$ (and hence $f(\pi)$) is very flat and consumption is not very sensitive to $\pi$.

\subsection{Connection with Abel Cubic Ode}
When $\gamma=-1/2$, the function in the RHS of \eqref{eq:h-deriv} is
a cubic polynomial with polynomial discriminant $D
=\frac{2}{27}(\rho + 0.5 r +\la_{10})-\frac{4}{27} (\rho + 0.5 r +
\la_{10})^2 (\la_{10} \hat{f})^2$. In a typical market we have $D<0$
and thus this cubic equation has one real root (namely,
$h_1$) and two complex conjugate roots (namely, $h_2$ and $\bar{h_2}$).
These roots can be obtained from the cubic formula.
Let $p = \Re(h_2)$ and $q = \Im(h_2)>0$ denote the real and imaginary parts of the complex
root $h_2$. Then
\begin{equation*}
\begin{split}
h_1=&\frac{\eta(-0.5)}{(\rho+0.5 r+\la_{10})} + \frac{9}{8\eta(-0.5)}, \qquad p = -\half h_1, \qquad q = \frac{\sqrt{3}}{2}\Bigl( h_1 - \frac{9}{4 \eta(-0.5)}\Bigr)
\end{split}
\end{equation*}
where
\begin{equation*}
\eta(-0.5)\triangleq\frac{3}{4} (\rho+0.5 r+\la_{10})^{2/3}\left(4 \la_{10}\hat{f}
+ \sqrt{(4\lambda_{10}\hat{f})^2-\frac{8}{(\rho+0.5
r+\la_{10})}}\right)^{1/3}.
\end{equation*}
Since \eqref{eq:h-deriv} is a separable ode, it can be again solved
explicitly using the above roots and partial fraction decomposition.
Furthermore, since the two complex roots are conjugate, we find the
real-valued solution to the Abel ode. To conclude, the implicit solution to the
ode \eqref{eq:h-deriv} for $\gamma=-0.5$ is given by
\begin{multline}\label{eq:g-abel}
\log \left\{\frac{|\phi(z)-h_1|}{
\sqrt{(\phi(z)-p)^2+q^2} }\right\}+\frac{h_1-p}{q}
\arctan \left(\frac{q}{\phi(z)-p}\right) \\
= -\frac{4}{27}((h_1-p)^2+q^2)(\rho+0.5r+\lambda_{10}) z.
\end{multline}
\begin{rem}
For the case of $\gamma=-1/3$, the degree of the polynomial ode
\eqref{eq:h-deriv} is four. The roots of a polynomial up to degree 4
can be explicitly found in terms of its coefficients. Using the
quartic formula, the polynomial ode \eqref{eq:h-deriv} of degree 4 can then also
be explicitly solved.
\end{rem}

\section{Numerical Examples}\label{sec:numeric}
To illustrate our derivations above we present a series of numerical examples. Recall that the model for $\gamma=0$ admits closed-form solutions given in Theorem \ref{thm:log} and for $\gamma=-1$ in \eqref{eq:f-hyperbolic}-\eqref{eq:coupled-hyberbolic}. These formulas are explicit up to finding the maximizer $\pi^*$ of the liquid regime that can be easily obtained through a standard numerical optimization algorithm. Formulas \eqref{eq:g-square-root} and \eqref{eq:g-abel} give implicit solutions for $W(\pi,x)$ for $\gamma=+0.5$ and $\gamma=-0.5$ respectively. Inverting these expressions numerically and plugging them into \eqref{eq:power-v-func} one can again straightforwardly solve the model. Finally, for general $\gamma$ one can use numerical integration routine on \eqref{eq:tildeg-soln} and combine it with a root-finding algorithm for \eqref{eq:power-v-func} to compute $V(x)$ and $W(\pi,x)$. The coupled system \eqref{eq:hjb-inf-horizon} can be solved similarly by employing a
  Picard iteration over successive numerical approximations to $V^0(x)$ and $V^1(\pi,x)$.

The key parameter in our model is the coefficient of risk-aversion $\gamma$. To illustrate its effect we focus on the three cases of $\gamma \in \{ 0, 0.5, -1\}$ which correspond to log-utility, square-root and hyperbolic risk-aversion coefficients. As discussed above, if $\hat{\pi}$ is moderate (less than $60-70\%$), the effect of liquidity freezes is typically negligible; therefore the interesting situation is where $\hat{\pi}$ is close to 1. Since $\hat{\pi}$ depends on $\gamma$ (cf.\ \eqref{eq:merton-pi-power}), to compare different HARA utilities we \emph{fix} $\hat{\pi}$ and then adjust the excess return $\mu-r$ as $\gamma$ varies. In other words we take $\mu = r + \hat{\pi}\sigma^2(1-\gamma)$.

Our base parameter set is $\rho = r = 0.05$, $\sigma=0.167$, $\alpha= 0.05$, $L=0$, $\la_{01} = 0.1$, $\la_{10} = 2$. The stock drift is $\mu = 0.05+0.025(1-\gamma)$, equivalent to $\hat{\pi} = 0.9$, which is rather large and makes the effect of illiquidity noticeable. This base case is supposed to represent a realistic stock market that experiences a major liquidity crisis about every ten years (e.g.\ historically in 1987, 1998, 2008, etc.); once started, the crisis lasts for about six months. To study the effect of liquidity shock parameters, we also vary the $\la$'s and $L$.

To study our asymptotic approximations for small $\la_{01}$, in Table
\ref{table:num} we compare the solution using the coupled system \eqref{eq:hjb-inf-horizon} to the approximate formulas \eqref{eq:pi-log-correction} ($\gamma=0$ case) and \eqref{eq:power-small-b} ($\gamma \neq 0$ case). The last two columns also compare the efficiency loss (relative to the classical Merton optimizer with $\la_{01}=0$) using the exact solution and the small-$\la_{01}$ approximation. As expected, when $\la_{01} \downarrow 0$, the optimal investment proportion $\pi^*$ converges to the Merton solution $\hat{\pi}$ and the quality of the asymptotic approximations improves. Same conclusions hold for the efficiency loss.

As can be seen in Table \ref{table:num}, for realistic parameter values and $L=0$, the efficiency loss is mild (less than 2\% for logarithmic and hyperbolic utilities and less than 1\% for square-root utility). This conclusion is consistent with the ``relaxed investor'' property of the Merton problem found by \cite{Rogers00,RogersZane} who studied efficiency loss when trading is discretized on a time grid. Thus, inability to rebalance continuously is not crucial, while cash crunch constraints will also be weak unless $\la_{10}$ is small.

Not surprisingly, if liquidity shocks are accompanied by significant price decreases, a major utility loss is experienced since the excess return of the asset is effectively cut. Even if a liquidity shock occurs once every 10 years and causes a relatively minor negative jump of $L=10\%$ in asset price, the optimal investment fraction $\pi^*$ in the risky asset is reduced from 90\% to 52\% for log-utility and 69\% for hyperbolic utility (the square-root investor is barely affected and down-weighs to 88.8\%).

In general, the derived asymptotic expressions underestimate the adjustment in risky investments $\pi^* < \hat{\pi} + \la_{01} \pi_1$ and overestimate the efficiency loss $\bt < \la_{01} \bt_1$. This is intuitive given the convexity/concavity of the  respective quantities and the linear approximations taken. Note that comparison of given quantities for different utilities is rather difficult, although we can claim that the square-root optimizer is highly insensitive to liquidity and the hyperbolic optimizer is somewhat more sensitive than the log-utility agent. Overall, Table \ref{table:num} seems to indicate that while the risk of negative asset jumps \emph{is} crucial, the risk of illiquidity \emph{per se} is not a significant driver of optimal investment decisions.


\begin{table}[ht]
$$\begin{array}{|l|cc|rr|}
\hline  \multicolumn{5}{|c|}{\text{Log Utility  } u(x) = \log(x), \mu = 0.075} \\
\text{Case} & \pi^* & \hat{\pi}+\la_{01}\pi_1 & \bt\; & \;\la_{01}\bt_1 \\ \hline
 \text{Base} & 0.879  &  0.846  &  1.08  &  1.15 \\
 \text{w/ }\la_{01} = 0.05 &  0.886  &   0.873 &   0.56  &  0.58 \\
 \text{w/ }\la_{01} = 0.02 & 0.893 &    0.889 &   0.23 & 0.23 \\ 
 \text{w/ }\la_{10} = 4 & 0.900 &     0.899 &    0.54 &  0.55 \\
\text{w/ }L = 0.1 & 0.521   & 0.467  & 13.97 &  18.13 \\
\text{w/ }\la_{01} = 0.5, \la_{10} = 10 & 0.900  &  0.900   & 1.06  &  1.11 \\ 
\hline \hline  \multicolumn{5}{|c|}{\text{Hyperbolic Utility  } u(x) = -x^{-1}, \mu = 0.1} \\
\text{Case} & \pi^* & \hat{\pi}+\la_{01}\pi_1 & \bt & \la_{01}\bt_1 \\ \hline
 \text{Base} & 0.857  &  0.683  &  1.92   & 2.36 \\
 \text{w/ }\la_{01} = 0.05 & 0.867   & 0.792  &  1.01  &  1.18 \\
 \text{w/ }\la_{01} = 0.02 & 0.879  &  0.857  &  0.43 & 0.47 \\ 
 \text{w/ }\la_{10} = 4 & 0.895   & 0.901  &  0.90  &  0.92 \\
\text{w/ }L = 0.1 & 0.691  &   0.652  &  17.16  & 18.55  \\
\text{w/ }\la_{01} = 0.5, \la_{10} = 10 & 0.900  &  0.916  &  1.77 & 1.83 \\ 
\hline \hline  \multicolumn{5}{|c|}{\text{Square-Root Utility  } u(x) = 2 \sqrt{x}, \mu = 0.0625} \\
\text{Case} & \pi^* & \hat{\pi}+\la_{01}\pi_1 & \bt & \la_{01}\bt_1 \\ \hline
 \text{Base} & 0.895   & 0.898  &  0.589  &  0.623 \\
 \text{w/ }\la_{01} = 0.05 & 0.897   & 0.899  &  0.302  &  0.312 \\
 \text{w/ }\la_{01} = 0.02 &  0.899 &   0.900  &  0.120  &  0.125 \\
 \text{w/ }\la_{10} = 4 & 0.898 & 0.900 &  0.304 & 0.314 \\
\text{w/ }L = 0.1 & 0.888  &  0.335  &  16.737  &  20.213 \\
\text{w/ }\la_{01}=0.5, \la_{10} = 10 & 0.878 &    0.838 &   0.606   & 0.638 \\
\hline \end{array}$$
\caption{Optimal investment proportions and efficiency loss for HARA utility-maximizers with $\gamma=0$ (log-utility), $\gamma=-1$ (hyperbolic utility) and $\gamma=0.5$ (square-root utility). $\bt$ measures the percent efficiency  loss compared to $\la_{01}=0$ (Merton problem) and $\bt_1$ is the asymptotic efficiency  loss derived in \eqref{eq:perc-loss-log} and \eqref{eq:perc-loss-power}.
\label{table:num}}
\end{table}

To further explore the role of investor's risk aversion, Figure \ref{fig:c1star} compares the optimal consumption schedules in the illiquid regime for three different levels of $\gamma$, namely $\gamma\in \{0.5,0,-0.5\}$. While comparison of value functions is hard due to different utility units, all consumption is proportional to total current wealth $x$ and moreover must decrease as the proportion of wealth in cash $(1-\pi)$ shrinks. Recall that in the classical Merton setting,
$c^*/x > \rho$ when $\gamma < 0$ and $c^*/x > \rho$ when $\gamma > 0$ and the proportion of wealth consumed is maximized at $\gamma = -1$.

\begin{figure}[ht]
\center{\includegraphics[height=3in]{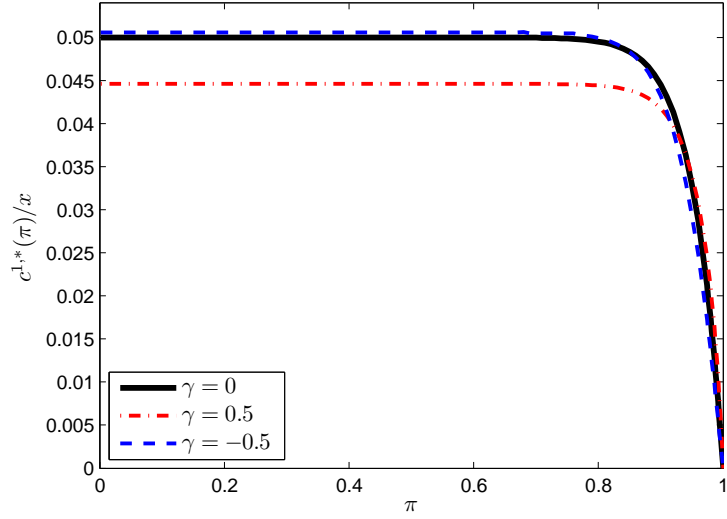}}
\caption{Consumption schedules $c^{1,*}$ in the illiquid regime. Parameter values are fixed at $\sigma = 0.166, \mu = 0.0625, r=\rho=0.05$ and $\la_{10}=1, L=0$. We compare the square-root ($\gamma=0.5$), logarithmic ($\gamma=0$) and inverse square-root ($\gamma=-0.5$) maximizers. All values are normalized by current wealth (i.e.\ we plot $c^{1,*}(\pi,x)/x$). \label{fig:c1star}}
\end{figure}

As shown in Figure \ref{fig:c1star}, $c^{1,*}$ is decreasing in $\gamma$ for moderate values of $\pi$ and is increasing in $\gamma$ for $\pi$ close to 1. This confirms our intuition that the cash crunch is most felt by the most risk-averse investors. Interestingly, around $\pi \sim 0.9$, no obvious ordering between $c^{1,*}(x; \gamma)$ is possible. For $\pi < 0.8$, $f'(\pi)$ is very small all $\gamma$'s, and therefore consumption is essentially constant in $\pi$.

\section{Homogenized Limit}\label{sec:homogenized}
We now consider the homogenized limit $\eps \to 0$ with respect to the transition matrix $Q^\eps = Q/\eps$ of $M$. This corresponds to the model where liquidity regimes change on a fast time-scale, with $\bar{\la}=\la_{10}/(\la_{01}+\la_{10})$ proportion of time spent in the liquid unconstrained regime. Formally, we assume that that the transition rates of $M$ are $\la_{10}/\eps$ and $\la_{01}/\eps$, the jumps are rescaled as $L = \eps\bar{L}$, and then proceed to carry out a first-order formal asymptotic expansion in the small parameter $\eps$.

As $\eps\to 0$, the effective duration of liquidity shocks goes to zero. Therefore, the investment proportion process $\pi_t$ converges  to a constant value, which is the optimizer of the equation defining $V^0(x)$.  Simultaneously, this means that $V^1(\cdot, x)$ becomes less and less sensitive to $\pi$, $\lim_{\eps \to 0} V^1_\pi(\pi,x) = 0$ for all $\pi \in [0,1)$ (super-polynomially), and therefore the optimization over $\pi$ in \eqref{eq:hjb-inf-horizon} is simplified.

As $\eps \to 0$, the ergodic Markov chain $M$ averages out the drift and volatility of $S$. Let us call a \emph{cycle} a pair of transitions of $M$ from state 0 to state 1 and back to state zero. Then the length of a cycle is about $\sim \eps(\la_{01}^{-1} + \la_{10}^{-1})$ during which the expected mean of an $S$-increment is $\sim \mu \eps \la_{01}^{-1} + \alpha \eps \la_{10}^{-1} - \eps \bar{L}$ and the quadratic variation of $S$ over one cycle is $\sim \eps \la_{01}^{-1} \sigma^2$. Simplifying the above expressions which become exact as $\eps \to 0$ and using the fact that the homogenized process remains a diffusion we obtain that the averaged asset dynamics $\bar{S}$ are
\begin{align*}
d\bar{S}_t/ \bar{S}_t = (\alpha +  (\mu - \alpha - \la_{01}\bar{L})\bar{\la} ) \, dt + \bar{\la}^{1/2} \sigma \, dW_t.
\end{align*}
However, as we will see below, the optimal strategy will not converge to the Merton solution corresponding to $\bar{S}$.

\subsection{Log Utility}
First, consider the log-utility investor.
We directly work with \eqref{eq:hjb-inf-horizon} and make the ansatz $V^0(x) = B_0 + \eps B_1 + \O(\eps^2) + \frac{1}{\rho} \log x$ and $V^1(\pi,x) = B_0 + \eps F(\pi) + \O(\eps^2)$. The leading term is the same in both equations because the dominant term in the expression for $V^1$ is $\la_{10}(V^0(x) - V^1(\pi,x))$. Collecting terms of order $\O(1)$ in the illiquid regime we obtain
$$
-\rho B_0 + \frac{1}{\rho}( r + (\alpha -r)\pi - \rho) + \log \rho = \la_{10}( F(\pi) - B_1).
$$
Back in the liquid regime, we find
$\pi^*_{hom} = \frac{\theta_{hom}}{2 \sigma}$, where \[
\theta_{hom} \triangleq \frac{1}{\sigma} ( \mu-r - \la_{01}\bar{L} + \frac{\la_{01}}{\la_{10}}(\alpha -r)  )\] and  $B_0$ solves
\begin{align}\label{eq:hom-limit-log}
B_0 = \frac{r/\rho -1 + \log \rho}{\rho} + \bar{\la}  \frac{\theta_{hom}^2}{2 \rho^2}.
\end{align}
Comparing with the benchmark in \eqref{eq:merton-inf-horizon-log}, we may interpret the limiting value function in \eqref{eq:hom-limit-log} as solution of the classical problem for a risky asset
with a modified squared Sharpe ratio of $\bar{\la} \theta_{hom}^2$ rather than $\theta^2$ as in \eqref{eq:merton-inf-horizon-log}. Therefore, the homogenized loss in utility computed in \eqref{eq:perc-loss-log} translates into
\begin{align}
\label{eq:perc-loss-fast-log}
\lim_{\eps\to 0} \bt \triangleq \bt^{(\log)}_{hom} & = 1 - \exp( \frac{1}{2 \rho} (\bar{\la}\theta_{hom}^2 - \theta^2)).
\end{align}
Counter-intuitively, we find that the derivative of $\bt^{(\log)}_{hom}$ with respect to $\la_{01}$ and $\la_{10}$ may be of either sign depending on parameter values. For instance, if $\mu > \alpha > r$ and $L = 0$ then $\partial\bt^{(\log)}_{hom}/\partial \la_{10} > 0$; however if $\mu$ is large and $\alpha < r$ then $\partial\bt^{(\log)}_{hom}/\partial \la_{10}  <0$.

In the special case $\alpha=r$, the explicit results of Proposition \ref{thm:log} apply and yield \eqref{eq:hom-limit-log} directly. In that case since the growth rate of the risky asset during liquidity shocks is $\alpha = r$, we can assign the illiquid regime a Sharpe ratio of zero.
 Then  $\bar{\la} \theta^2_{hom}$  corresponds exactly to averaging the squared Sharpe-ratios of each liquidity regime according to the invariant distribution of $M$. In contrast, the averaged price process $\bar{S}$ has excess return of $\bar{\la}(\mu-r - \la_{01}\bar{L})$ and volatility of $\sqrt{\bar{\la}}\sigma$. Thus, the illiquidity constraint causes the agent to apply the \emph{wrong} strategy over the averaged stock process. This is particularly clear in the case $L=0$ whereby $\lim_{\eps \to 0} \pi^* = \hat{\pi}$ which is the optimal strategy of the \emph{liquid} regime that is employed over the homogenized price process $\bar{S}$. In the related setting of a fast-scale regime-switching stock model with $M$-modulated drift and volatility (but no liquidity constraints), \cite{BauerleRieder04} showed that the log-investor also loses $\bt^{(\log)}_{hom}$ of utility compared to the Merton model. Therefore, for the log-investor in the homogenized limit the liquidity constraint completely disappears. When $\alpha \neq r$,
$\theta_{hom}$ no longer has an averaging interpretation and  the liquidity constraint remains in force even in the homogenized limit. For Example 1 in Section \ref{sec:numeric} above, we find $\bt^{(\log)}_{hom} = 1.066\%$, which is very close to the value $\Theta = 1.061\%$ for $\la_{01}=0.5, \la_{10}=10$ in Table \ref{table:num}.

\subsection{HARA Utility}

To obtain the homogenized limit for the HARA investor, we carry out a similar formal asymptotic expansion of equations \eqref{eq:f-power} and \eqref{eq:power-v-func} in the small parameter $\eps$.  Starting  once more from  the ansatz $V^0(x) = (B_0 + \eps B_1 + \O(\eps^2)) x^\gamma/\gamma$ and $V^1(\pi,x) = (B_0 +\eps F(\pi))x^\gamma/\gamma$, we obtain
\begin{align*}
\la_{10} (F(\pi) -B_1) = (-\rho + \gamma r + (\alpha -r)\pi)B_0 + (1-\gamma)B_0^{\frac{\gamma}{\gamma-1}},
\end{align*}
and the optimizer $\pi^*$ is
\begin{align*}
\pi^* = \frac{1}{(1-\gamma)\sigma^2} \left[ \mu - r - \la_{01}\bar{L} + \frac{\la_{01}}{\la_{10}\gamma} (\alpha -r) \right].
\end{align*}
Denoting $\theta^{(\gamma)}_{hom} \triangleq \frac{1}{\sigma} \left(\mu - r - \la_{01}\bar{L} + \frac{\la_{01}}{\la_{10}\gamma} (\alpha -r) \right)$, the equation for $B_0$ reduces to
\begin{align} \notag
 \bigl((-\rho+\gamma r) & B_0 + (1-\gamma)B_0^{\frac{\gamma}{\gamma-1}}\bigr)\left( 1+ \frac{\la_{01}}{\la_{10}} \right) + \frac{\gamma}{2(1-\gamma) \sigma} \left( \theta^{(\gamma)}_{hom} \right)^2 \cdot B_0 = 0, \\ \notag
 \Longrightarrow \quad B_0  & = (1-\gamma)^{1-\gamma}\left(\rho  -\gamma r - \bar{\la}(\theta^{(\gamma)}_{hom})^2 \frac{\gamma}{ 2(1-\gamma)} \right)^{\gamma-1}.
\end{align}
Again, we interpret the expression for $B_0$ as the solution of the classical Merton problem after modifying the Sharpe ratio of the risky asset to $\bar{\la}\theta^{(\gamma)}_{hom}$. The resulting efficiency loss is
\begin{align}
\bt^{(\gamma)}_{hom} = 1 - \left( \frac{ \rho - \gamma r - \bar{\la} \frac{\gamma (\theta^{(\gamma)}_{hom})^2}{ 2(1-\gamma)} }{ \rho - \gamma r - \frac{\gamma \theta^2}{2 (1-\gamma)} } \right)^{(\gamma-1)/\gamma}.
\end{align}

For Example I in Section \ref{sec:numeric} the efficiency losses for the homogenized limit are $\bt^{(-1)}_{hom} =1.742\%$ for Hyperbolic utility (compare with $1.766\%$ loss for $\la_{01} = 0.1, \la_{10} = 2$ in Table \ref{table:num}) and $\bt^{(0.5)}_{hom} = 0.600\%$ for Square-root utility (compare with $0.606\%$ loss in last row of Table \ref{table:num} for $\gamma=0.5$). Thus, these limiting expressions give accurate approximations even for moderate values of $\la_{01}$ and $\la_{10}$.

\section{Logarithmic Consumption on Finite Horizon}\label{sec:finite-horizon}
We proceed to study a finite horizon version of the model above. With finite horizon, time inhomogeneity introduces additional effects. In particular, close to terminal date $T$, both the cash crunch and the opportunity cost of illiquidity vanish.
The corresponding stochastic control problem under consideration is now:
\begin{align}\label{eq:finite-horizon-log}
\sup_{(\pi,c) \in \mathcal{A}} \E^x \left[ \int_0^T \e^{-\rho s}\log c_s \, ds \right].
\end{align}

\begin{prop}
For model \eqref{eq:finite-horizon-log}, the Merton solution is 
\begin{align}\label{eq:merton-finite-horizon-log}
\hat{V}(t,x) & \triangleq \hat{k}(t) \log x + \hat{h}(t) \\ \notag &
= \frac{1}{\rho}(1-\e^{-\rho(T-t)})\log x -\frac{1}{\rho}(1-\e^{-\rho(T-t)})\log
\left(\frac{1-\e^{-\rho(T-t)}}{\rho}\right)
\\& \notag \qquad\qquad +\left(\frac{(\mu-r)^2}{2\sigma^2 \rho^2}
+\frac{r-\rho}{\rho^2}\right)\left(1-\e^{-\rho
(T-t)}(1+\rho(T-t))\right)
\end{align}
The optimal consumption  and investment strategies are:
$$ \hat{c} = \frac{\rho x}{1-\e^{-\rho (T-t)}}, \qquad\qquad \hat{\pi} = \frac{\mu-r}{\sigma^2}.$$
\end{prop}

Assuming that $\la_{01}$ is small and with a slight abuse of notation, denote by $W(t,\pi,x)$ and $V(t,x)$ the first-order correction terms to $V^1(t,\pi,x)$ and $V^0(t,\pi,x)$ as in \eqref{eq:uncoupled-hjb}. Then the finite-horizon analogue of \eqref{eq:uncoupled-hjb} is now
\begin{align*}\left\{
\begin{aligned}
\sup_{c > 0} & \Bigl\{ W_t  + (x(r + (\alpha-r) \pi)-c) W_x + \pi(1-\pi)(\alpha-r) W_\pi + \pi \frac{c}{x} W_\pi \\ & \qquad + \lambda_{10}( \hat{V}(t,x) - W(t,\pi,x)) + \log c - \rho W\Bigr\}   = 0, \\
\sup_{\pi \in [0,1],c > 0} & \Bigl\{ V_t   + (x(r+(\mu-r) \pi) -
c)V_x + \half x^2 \pi^2 \sigma^2 V_{xx} \\ & \qquad +\la_{01}( W \big(t,
\frac{\pi (1-L)}{1-\pi L},(1-\pi L) x \big) - V(t,x)) + \log c - \rho
V\Bigr\} = 0,
\end{aligned}\right.
\end{align*}
with zero terminal conditions $W(T,\pi,x) = V(T,x) = 0$ for all $\pi,x$, and the singular boundary condition $W(t,1,x) = -\infty$ for all $t,x$.

\begin{thm}\label{thm:finite-horizon}
The value functions for \eqref{eq:finite-horizon-log} are $V(t,x) = \hat{k}(t) \log x + h^0(t)$ and $W(t,\pi,x) = \hat{k}(t) \log x + h^1(t,\pi)$ where $h^i$ satisfy
\begin{align}\label{eq:h-finite-horizon}
0 & = h^1_t -\rho h^1+r \hat{k}(t)  -1 + \la_{10}( \hat{h}(t) - h^1(t,\pi))
-\log(\hat{k}(t) - \pi h^1_\pi(t,\pi)),\\ \label{eq:h0-finite}
h^0(t) &= \int_t^T \e^{-(\rho+\la_{01})(s-t)}[r \hat{k}(s) -1 - \log \hat{k}(s)+
\xi(s) ] \, ds,
\end{align}
with terminal conditions $h^0(T) = 0$, $h^1(\pi,T) = 0$ and
\begin{equation} \label{eq:xi-finite}
\xi(t)\triangleq \sup_{\pi \in [0,1]} \Bigl\{ (\mu-r) \pi \hat{k}(t) - \half
\sigma^2 \pi^2 \hat{k}(t) +\la_{01} \cdot \left[\hat{k}(t) \log(1-\pi L)  + h^1
\big(t,\frac{\pi (1-L)}{(1-\pi L)} \big)\right] \Bigr\}.
\end{equation}
The small $\la_{01}$-asymptotics are
\begin{equation}\label{eq:small-la-finite} \left\{ \begin{aligned}
\pi^* &= \hat{\pi} +
\frac{\la_{01}}{\sigma^2}\left[\frac{-L}{1-\hat{\pi}L}+
\frac{1}{\hat{k}(t)} h^1_{\pi} (t, g(\hat{\pi}))g'(\hat{\pi})\right], \\
h^0(t) &= \hat{h}(t) + \la_{01} \Bigl[\int_t^T \e^{-\rho(s-t)}(h^1(s,g(\hat{\pi}))- \hat{h}(s))\,ds \\ & \qquad\qquad + \log(1-\hat{\pi}L) \left(\frac{1-\e^{-\rho(T-t)}(1+\rho(T-t))}{\rho^2}\right)
 \Bigr]. \end{aligned}\right.
\end{equation}
\end{thm}

Let us note that the key equation \eqref{eq:h-finite-horizon} is a first-order nonlinear pde on $[0,T] \times [0,1]$. Because of the cash-crunch, it has the degenerate boundary condition $\lim_{\pi \to 1} h^1(t,\pi) = -\infty$ which makes its numerical solution nontrivial. To overcome this difficulty, in Section \ref{sec:cash-crunch} below we carry out an asymptotic expansion of $h^1(t,\cdot)$ around $\pi=1$.

The small $\la_{01}$-asymptotics of the optimal investment fraction show that the effect of liquidity shocks mainly depends on stock jumps since $h^1_{\pi}$ is negligible for typical market parameter values and large time horizon. The correction term in the small $\lambda_{01}$-asymptotics of the value function $h^0$, however, depends on stock jumps as well as efficiency  loss during liquidity shocks. This correction term is always negative and decreases in magnitude as $t \to T$.

\subsection{Cash Crunch on Finite Horizon}\label{sec:cash-crunch}
When the market is illiquid and $\pi \simeq 1$, the agent experiences a cash crunch. The resulting singular boundary condition $V^1(t,1,x) = -\infty$ presents a numerical challenge for solving \eqref{eq:h-finite-horizon}. To resolve this issue,
we carry out an asymptotic analysis of the optimal behavior when $(1-\pi)$ is small.

\begin{lemma}
Let $h^1(t,\pi)$ be given in \eqref{eq:h-finite-horizon}. Then asymptotically for small $1-\pi$ and any $0 < d<1$
\begin{align}\label{eq:cash-crunch-asympt}
h^1(t,1-d\pi) = h^1(t,1-\pi) +\frac{1-\e^{-(\rho + \la_{10})(T-t)}}{\rho+\la_{10}} \log d + o(1-\pi).
\end{align}
\end{lemma}

\begin{proof}
When most of the wealth is in the stock and $M_t = 1$, consumption is driven by the amount of cash available, $y = (1-\pi)x$. Indeed, because the agent must maintain positive consumption, she concentrates on avoiding a complete cash depletion.
Let $\tau \sim Exp(\la_{10})$ be the next jump-time of $M$. Then the agent solves the optimization sub-problem
\begin{align*}
W(t,1-\pi,x) = C(y; x)  \triangleq \sup_{(c_t)} \E \left[ \int_0^{\tau \wedge (T-t)} \! \e^{-\rho s} \log c_s \,ds + \e^{-\rho(\tau \wedge (T-t))} V^0( \tau, X^x_{\tau}+Y^y_\tau) \right],
\end{align*}
where $(Y_t)$ is the amount of cash at date $t$ obeying $dY_t = (rY_t - c_t)\,dt$, and $X_{t}$ is the wealth in stock obeying $X_t = x\e^{ r t}$. Since $(1-\pi)$ is small, $X_\tau \gg Y_\tau$ and therefore consumption is asymptotically linear in remaining cash, since $c_s$ effectively only appears in the $\log c_s$-term. Therefore, for any fraction $d < 1$,
\begin{align*}
C(d \cdot y; x) & \simeq \int_0^{T-t} \e^{-\rho s} \log( d \cdot c^*_s) \e^{-\la_{10}s}  + \la_{10}\e^{-(\rho + \la_{10})s} V^0(s, x \e^{r s}) \, ds \\
& = C(y;x ) + \int_0^{T-t} \{ \e^{-(\rho + \la_{10})s} \log d \} \,ds \\
& = C(y;x) + \frac{1-\e^{-(\rho + \la_{10})(T-t)}}{\rho+\la_{10}} \log d.
\end{align*}
Re-arranging the last expression in terms of $h^1$ we obtain \eqref{eq:cash-crunch-asympt}.
\end{proof}

We may use \eqref{eq:cash-crunch-asympt} as a numerical boundary condition to solve \eqref{eq:h-finite-horizon}. Namely, we construct a finite grid $\{ 0, \Delta \pi, 2 \Delta \pi, \ldots, 1-\Delta \pi\}$ and solve \eqref{eq:h-finite-horizon} using a finite-difference scheme, applying
$$
h^1(t,1-\Delta \pi) = h^1(t,1-2\Delta \pi)- \frac{1-\e^{-(\rho + \la_{10})(T-t)}}{\rho+\la_{10}} \log 2,
$$
which is now a finite boundary condition at $\pi = 1-\Delta \pi$. Such a numerical solution is presented in Figure \ref{fig:h1} which compares the resulting $h^1(t,\pi)$ to the Merton solution $\hat{h}(t)$. As $t \to 0$, the percent efficiency  loss is linear in $\pi$ (becoming proportional in available cash); as $t \to \infty$, the efficiency loss converges to the infinite-horizon formula \eqref{eq:perc-loss-log}.

 \begin{figure}
 \centering{\includegraphics[height=3in]{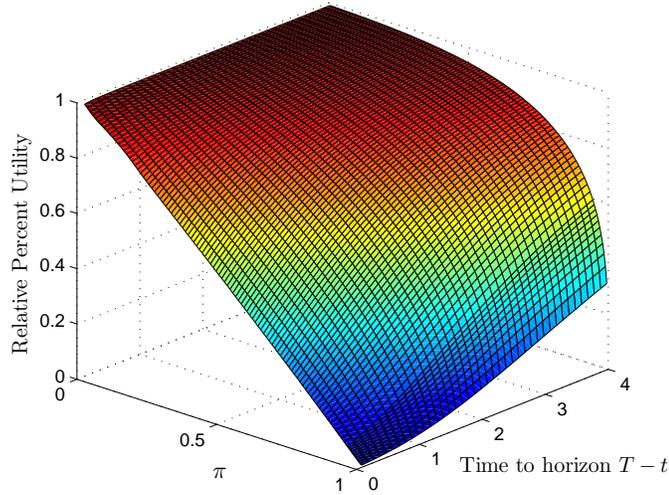}} \caption{Relative utility of $W(t,\pi,x)$ with respect to $\hat{V}(t,x)$ for different proportions of cash holdings $1-\pi$. On the vertical axis we plot the percent efficiency loss $\bt = 1 - \exp \left( \frac{ h^1(t,\pi) - \hat{h}(t) }{\hat{k}(t)} \right)$. Here $T=2, r=\rho=0.05, \mu=0.075, \sigma = 1/6$ and $L=0$.  \label{fig:h1}}
\end{figure}

\subsection{Terminal Utility without Consumption}
A related model recently appeared in \cite{Diesinger}, DKS hereafter. The latter paper was one of the inspirations for our analysis and considered optimization of terminal wealth under liquidity shocks. In contrast to our analysis, DKS did not allow for consumption but did allow a more general form of multiple liquidity shocks (rather than a two-state Markov chain). Similar to our results, DKS obtained explicit solutions for the case of log-utility. In this section we carry out the asymptotic analysis of their model which provides an instructive counterpart to our optimal consumption problem.

The DKS objective is
\begin{align}\label{eq:dks-objective}
\sup_{\pi \in \mathcal{A}} \E^{t,x} \left[ \log(X_T (1 - 1_{\{M_T = 1\}} \pi_T L) ) \right].
\end{align}
Thus, if the market is illiquid at terminal date $T$, total wealth is reduced by $L\%$.

First, we recall that
\begin{prop}
The Merton solution for maximizing terminal wealth is 
\begin{align}\label{eq:merton-log-terminal}
\sup_{\pi \in \mathcal{A}_0} \E^{t,x}[ \log (X_T)] = \log x + (r + \frac{(\mu-r)^2}{2\sigma^2})(T-t), \qquad \pi^* = \frac{(\mu-r)}{\sigma^2}.
\end{align}
\end{prop}

Denote by $J^0(t,x), J^1(t,\pi,x)$ the value functions corresponding to \eqref{eq:dks-objective}. Then $J^i$ solve the system
\begin{align*}
\left\{ \begin{aligned}
 J^{0}_t + \sup_{\pi \in [0,1]} \left\{ x(r + (\mu-r) \pi) J^{0}_x + \half x^2 \pi^2 \sigma^2 J^{0}_{xx}  + \lambda_{01}( J^{1}(t,x,\pi) - J^{0}(t,x)) \right\}& = 0, \\
 J^{1}_t + x(r+(\alpha-r) \pi) J^{1}_x + \lambda_{10} (J^{0}(t,x) - J^{1}(t,\pi,x)) + \pi(1-\pi) (\alpha-r) J^{1,k}_\pi &  = 0,\\
 J^{0}(T,x) = \log x, \qquad\text{and}\qquad J^{1}(T,\pi,x)  = \log x(1-\pi L). &
\end{aligned}\right.
\end{align*}

When $\alpha =r$, the solution reads \cite[Section 5]{Diesinger} $J^0(t,x) = \log x + f^0(t)$ and $J^1(t,\pi,x) = \log x + f^1(t,\pi)$ where
\begin{align}\label{eq:f0-dks}
f^0(t) & = \Bigl[ \left(r +(\mu-r)\pi^*_{DKS} - \half \sigma^2 (\pi^*_{DKS})^2 + \frac{\la_{01}}{\la_{10}}(r+\la_{10} \log(1-\pi^*_{DKS} L)) \right)\frac{1-\e^{-(\la_{01}+\la_{10})(T-t)}}{\la_{01}+\la_{10}} \\ \notag
& \qquad  -\frac{r}{\la_{10}}( \e^{-\la_{10}(T-t)} - \e^{-(\la_{01}+\la_{10})(T-t)}) \Bigr]\cdot \frac{\la_{01}}{\la_{01}+\la_{10}}  + \Bigl[ \frac{\la_{01}r }{\la_{10}^2} ( -1 + \e^{-\la_{10}(T-t)}) \\ \notag
& \qquad  + \left(r + (\mu-r)\pi^*_{DKS} - \half \sigma^2 (\pi^*)^2 + \frac{\la_{01}}{\la_{10}}(r+ \la_{10}\log(1-\pi^*_{DKS} L)) \right)(T-t) \Bigr]\cdot \frac{\la_{10}}{\la_{01}+\la_{10}} \\ \notag
f^1(t,\pi) & = \int_t^T (\la_{10} f^0(s) + (r+\la_{10}\log(1-\pi L) ) ) \e^{-\la_{10}(s-t)} \,ds + \log(1-\pi L) \e^{-\la_{10}(T-t)},
\end{align}
with optimal investment proportion $\pi^*_{DKS}$ defined through
\begin{align}\label{eq:pi-dks}
0 = (\mu-r) - \sigma^2 \pi^*_{DKS} - \la_{01} \frac{L}{1-L \pi^*_{DKS}}.
\end{align}

We proceed to understand the asymptotics of the above expressions when $\la_{01}$ is small. First, expanding \eqref{eq:pi-dks} in $\la_{01}$ we find (subject to the condition $1 > (\mu-r)/\sigma^2 > 0$)
\begin{align}\notag
\pi^*_{DKS} & = \frac{ \sigma^2 + (\mu-r) L - \sqrt{ (\sigma^2 - (\mu-r) L)^2 + 4 \sigma^2 L^2 \la_{01}}}{2 \sigma^2 L} 
\\ & =  \frac{\mu-r}{\sigma^2} - \la_{01} \frac{L}{ \sigma^2 - (\mu-r) L} + \mathcal{O}(\la_{01}^2). \label{eq:pi-star-dks-small}
\end{align}
As expected, risk of illiquidity reduces holdings of stock, with the effect increasing as (i) $\sigma$ decreases or (ii) $\mu$ increases or (iii) the liquidity jump loss $L$ increases. For typical parameter values the denominator above is quite small and therefore the relative changes in $\pi^*_{DKS}$ can be quite big even for a small change in the underlying parameter.
Note that if $L = 0$ then the optimal strategy is always equal to the Merton solution $\hat{\pi}$, due to the form of the first-order condition of $\pi^*_{DKS}$. This phenomenon of myopic investment is unique to log-utility. In the same vein, $\pi^*_{DKS}$ is in fact independent of $\la_{10}$ to first order.

Comparing the expression \eqref{eq:pi-star-dks-small} of the optimal investment proportion to a similar expression \eqref{eq:pi-log-correction} derived in Section \ref{sec:log-utility} we see the impact of consumption on investment. Indeed, the two problems are identical except that intermediate consumption introduces the cash-crunch feature which increases the impact of illiquidity. Analytically, this difference is exactly represented by
$$
\pi^*_{DKS} - \pi^* = \la_{01} \frac{1-L}{(1-\hat{\pi}L)^2} \frac{ g(\hat{\pi})^{\la_{10}/\rho}}{1-g(\hat{\pi})^{\la_{10}/\rho}}.
$$
As already discussed, unless $\hat{\pi}$ is very close to 1 or $\la_{10}$ is small, this difference is relatively small and therefore there is little difference between maximization of terminal wealth and maximization of utility of consumption.

Plugging into the expression for $f^0$ in \eqref{eq:f0-dks} and simplifying we end up with
\begin{align*}
f^0(t) & = (r + \frac{(\mu-r)^2}{2\sigma^2})(T-t) + \la_{01}\Bigl[ \frac{(\mu-r)^2}{2 \sigma^2\la_{10}^2}\cdot \left(1-\e^{-\la_{10}(T-t)}-\la_{10}(T-t) \right)  \\ & \qquad +\log \left(1- \frac{(\mu-r) L}{\sigma^2} \right)(T-t) \Bigr] + \mathcal{O}(\la_{01}^2).
\end{align*}
Comparing with the limiting case $\la_{01} =0$ in \eqref{eq:merton-log-terminal}, we obtain that the percentage loss in utility is $\bt^{(DKS)}(t) = \la_{01} \bt_1^{(DKS)}(t) + \O(\la_{01}^2)$ with
\begin{align}
\bt_1^{(DKS)}(t) = \Bigl[ \frac{(\mu-r)^2}{2 \sigma^2\la_{10}^2}\cdot \left(\e^{-\la_{10}(T-t)}+\la_{10}(T-t)-1 \right)  - \log \left(1- \frac{(\mu-r) L}{\sigma^2} \right)(T-t) \Bigr].
\end{align}
Both terms above are always positive and decrease (approximately linearly) in magnitude as $t \to T$.
Again, the squared Sharpe ratio makes an appearance. Counterintuitively, as $\la_{10} \to \infty$, the illiquidity correction does not vanish, even if the amount of time spent in the illiquid regime is negligible. This is because there is always a possibility of an illiquid shock right before $T$, proportional to $\la_{01}$ which would produce a loss of $L\%$ of equity.

Similar analysis can be done for the fast-scale liquidity freezes:
\begin{lemma}\label{lem:fast-dks}
Rescale the generator of $M$ as $Q^\eps = \eps^{-1} Q$. Then as $\eps \to 0$ we have
\begin{align*}
\pi^* & \to \frac{(\mu-r-\la_{01}\bar{L})}{\sigma^2},\qquad
f^0(t) \to (T-t)r + \bar{\la}(T-t) \frac{ (\mu-r-\la_{01}\bar{L})^2}{\sigma^2}.
\end{align*}
Therefore the homogenized percent efficiency loss is
\begin{align}\label{eq:perc-loss-dks}
\bt^{(DKS)}_{hom}(t) = 1-\exp \left( (T-t) \frac{1}{2} \{ \bar{\la}\tilde{\theta}^2 - \theta^2\} \right).
\end{align}
\end{lemma}

We omit the proof of Lemma \ref{lem:fast-dks} which is similar to the computations for \eqref{eq:perc-loss-fast-log} and follows from expressions \eqref{eq:pi-dks}-\eqref{eq:f0-dks} after tedious arithmetic simplifications. Also note that the efficiency loss in \eqref{eq:perc-loss-dks} essentially matches \eqref{eq:perc-loss-fast-log} except for the discount factor.

\section{Conclusion}\label{sec:conclusion}
In this paper we considered a perturbed version of the classical Merton optimal investment-consumption model. Our modification accounts for possibility of liquidity shocks during which the agent cannot trade and cannot use her stock holdings for consumption. The liquidity shocks also lead to negative jumps in risky asset price and reduce its excess return. We obtained explicit formulas for HARA utility-maximizers, especially for the special cases of $\gamma=0, \pm 0.5, -1$. The economic interpretation of our results was carried out by deriving the percent efficiency loss relative to the respective Merton problem solution. We showed that the asymptotic effect of illiquidity can be captured by appropriately modifying the Sharpe ratio of the underlying asset inside the classical formulas. This gives a simple rule to the investor in terms of accounting for liquidity shocks. Overall, we found that for realistic parameter values, the efficiency loss is primarily driven by the jump
 parameter $L$.

Our model can be easily extended if one is willing to sacrifice analytical tractability. Natural extensions would allow for more complex stock dynamics (such as general jump-diffusion processes, non-zero volatility in the illiquid regime), arbitrary utility functions, further market frictions or multiple stocks. The resulting HJB equations would then be modified version of the respective classical optimal investment models, see e.g.\ \cite{KarShreveMethods}. Also, other risk objectives might be considered to judge the attitude of the investor towards illiquidity. Our first setting where probability of a liquidity shock is small resembles the lifetime investment problem, in which case quantities such as probability of ruin become important risk metrics. Another direction that we leave to future work is to consider the effect of liquidity shocks on indifference pricing of contingent claims.

\section*{Acknowledgment}
We thank  Mihai Sirbu and the participants of the 2010 Young Researchers' Workshop on Finance (Tokyo) for their feedback on preliminary versions of this paper.

\bibliography{liquidityAsympt}
\bibliographystyle{jas99}

\appendix
\section*{Proof of Proposition \ref{thm:verify}}
\begin{proof}

The proof is largely standard and follows the classical Verification Theorem for stochastic control of diffusions \cite[Theorem IV.5.1]{FlemingSoner}. To simplify the presentation, we consider the case $u(0) > -\infty$. In the sequel we will primarily work with HARA utilities $u(x) = x^\gamma/\gamma$ which are infinite at the origin for $\gamma <0$. That case can be treated similarly to the proof below by using the resulting homothetic property of the value function (see Lemma \ref{lem:hom-degree-gamma} and the methods in \cite{ShreveSoner94}). Let us also mention the paper \cite{SotomayorCadenillas09} which had a similar verification theorem for optimal consumption in a regime-switching model.

Define $\tilde{u}(x) = u(x)-u(0)$. Then $\tilde{u}$ is another utility function, and
$ \E[ \int_0^\infty \e^{-\rho t} \tilde{u}(c_t) \,dt] = \E[ \int_0^\infty \e^{-\rho t} u(c_t) \,dt] -u(0)/\rho$. Moreover the admissible strategy sets are the same for the problem with $\tilde{u}(\cdot)$ and $u(\cdot)$. Therefore, without loss of generality we assume for the remainder of the proof that $u(0)=0$.

To apply the conclusions of \cite[Theorem IV.5.1]{FlemingSoner}, the following steps need to be verified.

{\it Step 1:}
Check that $V^i(\cdot) < \infty$ are finite. Case 1: $\sup_c u(c) = +\infty$. In this situation of ``positive risk aversion'', it may well be the case that the value functions are infinite since even a modest level of consumption can generate a lot of utility; see Proposition \ref{lem:merton-inf-hara} below. To avoid this we assume that the liquidity-unconstrained problem with value functions
$$
\tilde{V}^0(x) = \sup_{\pi,c \in \mathcal{A}_0} \E \left[ \int_0^\infty \e^{-\rho t} u(c_t(X^{x,\pi,c}_t)) \, dt \,| \, M_0 = 0 \right],
$$
and similarly for $\tilde{V}^1(\pi,x)$, are finite. We refer to \cite{SotomayorCadenillas09} for exhaustive analysis of that setting.
Since $\mathcal{A} \subseteq \mathcal{A}_0$ we immediately obtain that under the above assumption, $V^0(x) \le \tilde{V}^0(x) < \infty$ and  $V^1(\pi,x) \le \tilde{V}^1(\pi,x) < \infty$.

Case 2: $\sup_c u(c) <\infty$. In that case, we immediately have $V^1(\pi,x) \le V^0(x) \le \{ \sup_c u(c) \} \rho^{-1}$.

{\it Step 2:}
We next show that $J^i(\cdot)$ solving \eqref{eq:hjb-inf-horizon} are  upper bounds for the value functions.

Let $(\pi,c)$ be an admissible strategy. The assumptions on $(\pi,c)$ then guarantee that the sde in \eqref{eq:x-dynamics} has a unique strong solution (in fact, a closed-form representation is possible, see \cite{SotomayorCadenillas09}).

The basic idea is now to apply Ito's lemma to the semi-martingale $(X^{x,\pi,c}_t,M_t)$. Let $\theta \triangleq \inf \{t : X^{x,\pi,c}_t = 0 \}$. Since wealth is required to stay non-negative, no consumption is possible on $(\theta,\infty]$ and it follows that $V^0(0) = V^1(\pi,0) = 0$.
Fix initial wealth $x$ and constants $0 < a_1 < x < a_2 < \infty$ and $a_3 < 1$. Let $\tau \triangleq \tau_{1} \wedge \tau_2 \wedge \tau_3 \wedge \theta$, with
$$ \tau_1 \triangleq \inf \left\{ t: X^{x,\pi,c}_t \le a_1 \right\}, \quad \tau_2 \triangleq \inf \left\{ t: X^{x,\pi,c}_t \ge a_2 \right\}, \quad \tau_3 \triangleq \inf \left\{ t : \{\pi_t \ge a_3\} \cap \{M_t = 1\} \right\}.$$
For any $t < \infty$, the quantities $X^{x,\pi,c}$, $J^0(X^{x,\pi,c}_s)$, $ J^0_x(X^{x,\pi,c}_s)$, $J^1(\pi_s, X^{x,\pi,c}_s)$, being solutions of \eqref{eq:hjb-inf-horizon}, are local martingales. Moreover, they are bounded on  $[0,t \wedge \tau]$ and consequently are all true martingales until $t\wedge \tau$.  Applying the optional sampling theorem we obtain
\begin{align}\label{eq:remainder-term-hjb}
J^0(x) - \E \left[ \int_0^{t\wedge \tau} \e^{-\rho s} u(c_s) \,ds \right] \ge \E \Bigl[ \e^{-\rho t \wedge \tau} (J^0(X^{x,\pi,c}_{t\wedge \tau}) 1_{\{M_{t\wedge \tau} = 0\}} + J^1(\pi_{t\wedge \tau},X^{x,\pi,c}_{t\wedge \tau}) 1_{\{M_{t\wedge \tau} = 1\}} \Bigr].
\end{align}
Now the right hand side in \eqref{eq:remainder-term-hjb} is non-negative, so that $J^0(x) \ge \E^x \bigl[ \int_0^{t\wedge \tau} \e^{-\rho s} u(c_s) \,ds \bigr]$ and taking $a_1 \to 0, a_2 \to \infty, a_3 \to 1$ and using Monotone Convergence Theorem (MCT), we obtain $J^0(x) \ge \E^x \bigl[ \int_0^{t \wedge \theta} \e^{-\rho s} u(c_s) \, ds \bigr]$. Finally, taking $t \to \infty$ and again using MCT we find $J^0(x) \ge V^0(x; (\pi,c))$ as desired. Similar steps apply to $J^1(\pi,x)$.

{\it Step 3:} Verify existence of optimal controls. In our case sufficient first order conditions for the consumption optimizer above imply that
\begin{align*}
c^{1,*}  = (u')^{-1}(V^1_x - \frac{\pi}{x} V^1_\pi), \qquad
c^{0,*} = (u')^{-1}(V^0_x),
\end{align*}
Assumed smoothness of $V^i$ guarantees existence of $c^{i,*}$. The investment optimizer solves
$$
(\mu-r)x V^0_x + \pi^* x^2 \sigma^2 V^0_{xx} + V^1_\pi \left( g(\pi^*), (1-\pi^* L)x \right)g'(\pi^*) - L x V^1_x \left(g(\pi^*), (1-\pi^* L)x \right) = 0,
$$
and an interior solution on $[0,1]$ is guaranteed if $\mu > r$. Indeed, in the latter case the above expression is positive when $\pi=0$ and goes to $-\infty$ as $\pi \to 1$ (we have $u'(0) = +\infty$ and as proportion of liquid wealth shrinks, consumption must be reduced to zero, implying that  $\lim_{\pi \uparrow 1} V^1_\pi = -\infty$).

\end{proof}

\section*{Proof of Lemma \ref{lem:power-ode-soln}}
\begin{proof}
The ode \eqref{eq:h-deriv} is separable and therefore the corresponding theory may be applied \cite[p.31]{WalterBook}. We begin by re-writing in terms of indefinite integrals as
\begin{align}
\int dz + C = \int \frac{d \phi}{{\gamma}(\rho - \gamma
r + \la_{10}) \left(\frac{\phi}{(1-\gamma)}\right)^{1-1/\gamma}-\gamma\phi-
\gamma\la_{10} \hat{f}}.
\end{align}

First consider the case $\gamma>0$. An easy argument shows that $\phi(\cdot)$ is increasing and therefore the solution satisfying $\phi(0) = 0$ is given implicitly by
\begin{align}
z = \int_0^{\phi(z)}  \!\!\frac{d x}{{\gamma}\tilde{H}(x)}, \qquad \tilde{H}(x) = \left\{(\rho - \gamma r +\la_{10}) \left( \frac{x}{1-\gamma} \right)^{1-1/\gamma} - x - \la_{10}\hat{f} \right\}.
\end{align}
Recall the horizontal asymptote $\lim_{z \to \infty} \phi(z) = (1-\gamma)f(0)^{\frac{\gamma}{\gamma-1}}$. Checking that $\tilde{H}((1-\gamma)f(0)^{\frac{\gamma}{\gamma-1}}) =H(f(0)) = 0$, and using
$\tilde{H}(x)>0$ for $x<f(0)$  and $\gamma>0$ we see that $\lim_{z \uparrow \phi(\infty)}\int_0^z \frac{dx}{\gamma \tilde{H}(x)} =+\infty$ and thus the boundary conditions are matched correctly.

When $\gamma < 0$, $\phi(z)$ is decreasing; therefore for $N > \phi(\infty)$ the unique solution to \eqref{eq:h-deriv} satisfying $\phi(0) = N$ is given by
\begin{align}\label{eq:int-neg-gamma}
z = -\int_{\phi(z)}^N \frac{d x}{\gamma \tilde{H}(x)}.
\end{align}
Moreover, since the continuous function $\tilde{H}$ is bounded on the compact domain $[\phi(\infty),N]$, the solution $\phi(\cdot)$ is Lipschitz on $[0,\infty)$.
Since the leading power in $\tilde{H}(x)$ is $1-\gamma^{-1}>1$ for negative $\gamma$, the integral in \eqref{eq:int-neg-gamma} is proper and converges as $N\to \infty$. Therefore, taking the limit, we find that the unique solution with $\lim_{z \to 0}\phi(z) = +\infty$ is
\begin{align}\label{eq:tildeg-soln}
z = \int_{\phi(z)}^\infty \frac{-d x}{\gamma\tilde{H}(x)}.
\end{align}
Fix $N'$ large enough. Then for $x> N'$, $\gamma \tilde{H}(x) = \gamma (\rho - \gamma r +\la_{10})\left(\frac{x}{1-\gamma}\right)^{1-1/\gamma} + o(x^{1-1/\gamma})$ and therefore
$$
\int_{N'}^\infty \frac{-d x}{{\gamma}\tilde{H}(x)} \simeq \int_{N'}^\infty \frac{1}{A} x^{1/\gamma-1} dx  = -\frac{\gamma}{A} (N')^{1/\gamma}, \qquad A\triangleq (\rho - \gamma r +\la_{10}) (1-\gamma)^{\gamma^{-1}-1} + o(N').
$$
Comparing with \eqref{eq:tildeg-soln} we find that
$$
\phi(\eps) = \eps^\gamma (1-\gamma)^{1-\gamma}(\rho - \gamma r +\la_{10})^\gamma + \O(\eps).
$$
Thus, $\phi$ grows polynomially (at rate $\gamma$) around zero. Finally, transforming back into the $\pi$-variable, and using that for $\eps$ small enough, $\e^{-\eps} > 1-2 \eps$ we obtain
$$
f(1-\eps) = \O( \eps^\gamma) 
$$
showing that $f$ also grows polynomially as the cash crunch is approached.
\end{proof}

\section{Proof of Theorem \ref{thm:finite-horizon}}
\begin{proof}
Once again following the form of the unperturbed problem, we posit
that under liquidity shocks, the value functions are $V(t,x) =
k^0(t) \log x + h^0(t)$, $W(t,\pi,x) = k^1(t) \log x + h^1(t,\pi)$.
Substituting the ansatz for the illiquid regime value function $W$
we obtain that $c^{1,*} = \frac{x}{k^1(t)- \pi h_{\pi}^1(t,\pi)}$
and after some simplifications that $k^1 \equiv \hat{k}$. Moreover, the
wealth-independent term $h^1(t,\pi)$ solves \eqref{eq:h-finite-horizon}.

Returning to the liquid regime, the ansatz implies that $c^{0,*} = x/k^0(t)$ and
\begin{multline*}
k^0_t \log x  + h^0_t  + r k^0(t) - 1 + \la_{01}\log
x[k^1(t)-k^0(t)] -\la_{01} h^0(t) \\+ \log x - \log k^0(t)
-\rho[k^0(t) \log x + h^0(t)]
 +  \xi(t,\pi^*)= 0.
\end{multline*}
where
\begin{equation*}
\xi(t,\pi^*)=\sup_{\pi \in [0,1]} \Bigl\{ (\mu-r) \pi k^0(t) - \half
\sigma^2 \pi^2 k^0(t) +\la_{01} \left[k^1(t) \log(1-\pi L)  + h^1
\bigl(t, g(\pi) \bigr)\right] \Bigr\}.
\end{equation*}
The terminal condition is $k^0(T) = 0, h^0(T) = 0$. Considering the terms multiplying $\log x$, we find again
$k^0 \equiv \hat{k}$. Thus optimal consumption in the liquid regime
$c^{0,*}$ is again unaffected by liquidity shocks. Similarly, collecting terms of $\O(1)$, we see that $h^0$ solves
the first-order ode
\begin{equation*}
h^0_t-\rho h^0(t)-\la_{01} h^0(t) + r k^0(t)-1-\log k^0(t) +\xi(t,\pi^*)=0.
\end{equation*}
Since this first-order ode is linear and $\xi(t,\pi^*)$ is independent
of $h^0$, we can apply variation of constants to obtain \eqref{eq:h0-finite} and \eqref{eq:xi-finite}.
Next, from equation \eqref{eq:xi-finite} we know the optimal investment fraction
$\pi^*$ satisfies
\begin{equation}\label{eq:pi-finite-log}
(\mu-r) - \sigma^2 \pi^* + \la_{01} \left[\frac{-L}{1-\pi^*L}+
\frac{1}{k^0(t)} h^1_{\pi} (t, g(\pi^*))g'(\pi^*)\right]=0.
\end{equation}
Using equation \eqref{eq:pi-finite-log} we expand $\pi^*$ to the
leading order for small $\la_{01}$ to obtain the first half of \eqref{eq:small-la-finite}.
Plugging-in this asymptotic expression  for $\pi^*$ into \label{eq:h0}
and expanding \eqref{eq:h0-finite} for small $\la_{01}$ to the leading
order, we recover the second half of \eqref{eq:small-la-finite}.
\end{proof}

\end{document}